\documentclass[runningheads]{llncs}

\usepackage{orcidlink}

\usepackage[T1]{fontenc}

\usepackage{graphicx}
\usepackage{color}
\definecolor{darkgreen}{rgb}{0.0, 0.6, 0.0}
\definecolor{darkred}{rgb}{0.9, 0.0, 0.0}
\definecolor{darkblue}{rgb}{0.0, 0.0, 0.9}

\usepackage{xcolor}

\usepackage{tcolorbox}

\usepackage{hyperref}
\hypersetup{
    colorlinks,
    linkcolor={blue},
    citecolor={blue},
    urlcolor={blue}
}

\usepackage{stmaryrd}
\usepackage{mathrsfs}
\usepackage{amsmath,amssymb,eurosym}
\usepackage{xspace}

\usepackage{tikz}
\usetikzlibrary{automata}
\usetikzlibrary{arrows}
\usetikzlibrary{arrows.meta} 
\usetikzlibrary{positioning}
\usetikzlibrary{calc}
\usetikzlibrary{shapes}
\usetikzlibrary{graphs}

\usepackage{thm-restate}

\usepackage[shortlabels]{enumitem}

\usepackage{wrapfig}

\usepackage{csquotes}

\usepackage[labelfont=bf]{caption}
\usepackage{subcaption}
\usepackage{wrapfig}

\usepackage[capitalise]{cleveref}
\Crefformat{figure}{#2Fig.~#1#3}
\Crefmultiformat{figure}{Figs.~#2#1#3}{ and~#2#1#3}{, #2#1#3}{ and~#2#1#3}

\newcommand{\faultClass}[1][]{{\color{red!50!black}{\texttt{f}_{#1}}}}

\newcommand{\LTLThree}{\ensuremath{\text{LTL}_3}\xspace}
\newcommand{\LTLFour}{RV-LTL\xspace}
\newcommand{\BoolThree}{\mathbb{B}_3}

\newcommand{\currentlyFalse}{\ensuremath{\texttt{f}^\mathsf{p}}}
\newcommand{\currentlyTrue}{\ensuremath{\texttt{t}^\mathsf{p}}}

\ProvideDocumentCommand{\set}{m}{\left\{\, #1 \,\right\}}
\ProvideDocumentCommand{\where}{}{\mskip6mu\middle|\mskip6mu}
\ProvideDocumentCommand{\filter}{mm}{\set{#1 \where #2}}

\ProvideDocumentCommand{\abs}{m}{{\left| #1 \right|}}

\ProvideDocumentCommand{\tuple}{m}{{\left\langle #1 \right\rangle}}

\DeclareDocumentCommand{\mxDeclareMathSymbol}{m m O{} O{}}{
  \DeclareDocumentCommand{#1}{t{'} D||{} O{} O{}}{
    \ensuremath{\smash{\mathord{#3 {#2}\IfBooleanT{##1}{'}##2^{##4}_{##3} #4}}}
  }
}
\DeclareDocumentCommand{\mxDeclareMathBin}{m m O{} O{}}{ %
  \DeclareDocumentCommand{#1}{t{'} D||{} O{} O{}}{ %
    \ensuremath{\mathbin{#3 {#2}_{##3}\IfBooleanT{##1}{'}##2^{##4} #4}} %
  } %
}

\DeclareDocumentCommand{\mxDeclareMathRel}{m m O{} O{}}{ %
  \DeclareDocumentCommand{#1}{t{'} D||{} O{} O{}}{ %
    \ensuremath{\mathrel{#3 {#2}_{##3}\IfBooleanT{##1}{'}##2^{##4} #4}} %
  } %
}

\newcommand{\mxPreFunSkipBack}{ %
  \mathchoice %
  {\mskip-4mu} %
  {\mskip-4mu} %
  {\mskip-2mu} %
  {\mskip-0mu} %
}

\DeclareDocumentCommand{\mxDeclareMathFun}{m m O{} O{} O{} O{} D||{(} D||{)}}{ %
  \DeclareDocumentCommand{#1}{t{'} D||{} O{} O{} d()}{ %
    \ensuremath{ %
      \mathord{#3 {#2}_{#5 ##3}\IfBooleanT{##1}{'}##2^{#6 ##4} #4} %
      \IfNoValueF{##5}{\mxPreFunSkipBack\left#7 ##5 \right#8} %
    } %
  } %
}

\newcommand{\mathsc}[1]{\text{\normalfont\fontfamily{lmr}\selectfont\textsc{#1}}}

\mxDeclareMathFun{\fTrace}{\sigma}

\mxDeclareMathFun{\fPowerSet}{\wp} 
\mxDeclareMathFun{\fDm}{\mathcal{V}_{\mathrm{det}}} 

\mxDeclareMathSymbol{\aWord}{\sigma}
\mxDeclareMathSymbol{\anInfSeq}{\sigma}

\mxDeclareMathFun{\bigO}{\mathcal{O}}
\mxDeclareMathFun{\fBigO}{\bigO}

\mxDeclareMathSymbol{\ltlAlways}{\square}
\mxDeclareMathSymbol{\ltlEventually}{\lozenge}
\mxDeclareMathSymbol{\ltlNext}{\bigcirc}

\mxDeclareMathSymbol{\sAtomicProps}{\mathrm{AP}}
\mxDeclareMathSymbol{\anAtomicProp}{a}
\mxDeclareMathSymbol{\anAPSet}{A}
\mxDeclareMathSymbol{\aBoolExpr}{\phi}
\mxDeclareMathSymbol{\boolExprTrue}{\texttt{true}}
\mxDeclareMathSymbol{\boolExprFalse}{\texttt{false}}
\mxDeclareMathRel{\boolAnd}{\land}
\mxDeclareMathRel{\boolOr}{\lor}
\mxDeclareMathRel{\boolImplies}{\rightarrow}
\mxDeclareMathSymbol{\boolNot}{\lnot}
\mxDeclareMathFun{\fBoolExprs}{\mathbb{B}}|[||]|
\mxDeclareMathFun{\fBoolSatis}{}|\llbracket||\rrbracket|

\mxDeclareMathFun{\fVerdict}{\mathsf{v}}

\mxDeclareMathSymbol{\aSet}{X}
\mxDeclareMathSymbol{\aSubset}{Y}
\mxDeclareMathSymbol{\anOtherSet}{\mathbf{Y}}
\mxDeclareMathSymbol{\anElement}{x}
\mxDeclareMathSymbol{\emptySeq}{\epsilon}

\mxDeclareMathRel{\join}{\sqcup}
\mxDeclareMathRel{\meet}{\sqcap}

\mxDeclareMathSymbol{\beliefDomain}{\mathfrak{B}}
\mxDeclareMathSymbol{\sVerdicts}{\mathsf{V}}
\mxDeclareMathSymbol{\aVerdict}{v}
\mxDeclareMathRel{\isMoreSpecific}{\sqsubseteq}
\mxDeclareMathRel{\beliefRefines}{\sqsubset}

\mxDeclareMathSymbol{\anLTS}{\mathfrak{T}}
\mxDeclareMathSymbol{\sLabels}{\mathbf{\Lambda}}
\mxDeclareMathFun{\fLabel}{\lambda}
\mxDeclareMathSymbol{\sAlphabet}{\mathbf{\Sigma}}
\mxDeclareMathSymbol{\aSymbol}{a}
\mxDeclareMathSymbol{\aSymbolSet}{\Sigma}
\mxDeclareMathSymbol{\sStates}{\mathbf{S}}
\mxDeclareMathSymbol{\aState}{s}
\mxDeclareMathSymbol{\aStateSet}{S}
\mxDeclareMathSymbol{\transitionRel}{T}
\mxDeclareMathSymbol{\initialState}{\bar{s}}
\mxDeclareMathFun{\fPost}{\mathsc{Post}}
\mxDeclareMathFun{\fReach}{\mathsc{Reach}}
\mxDeclareMathFun{\fEnabled}{\mathsc{Enabled}}

\mxDeclareMathFun{\fLang}{\mathcal{L}}

\mxDeclareMathRel{\ltlUntil}{\mathsf{U}}

\mxDeclareMathFun{\fExecs}{\mathsf{E}}

\mxDeclareMathSymbol{\aVTS}{\mathcal{V}}
\mxDeclareMathSymbol{\sControlStates}{\mathcal{Q}}
\mxDeclareMathSymbol{\aControlStateSet}{Q}
\mxDeclareMathSymbol{\aControlState}{q}
\mxDeclareMathFun{\fBeliefs}{\mathsc{Verdicts}}
\mxDeclareMathFun{\fLower}{\mathsc{Lower}}

\mxDeclareMathSymbol{\sFeatures}{\mathcal{F}}
\mxDeclareMathSymbol{\aFeatTS}{\anLTS_{\mathrm{feat}}}

\mxDeclareMathSymbol{\aWord}{\hat{\sigma}}
\mxDeclareMathSymbol{\sObservables}{\mathfrak{O}}
\mxDeclareMathSymbol{\anObservable}{a}

\providecommand{\mc}[1]{{\mathcal{#1}}}

\mxDeclareMathSymbol{\Real}{\mathbb{R}}
\mxDeclareMathSymbol{\Rational}{\mathbb{Q}}
\mxDeclareMathSymbol{\Nat}{\mathbb{N}}

\mxDeclareMathSymbol{\aConstant}{c}

\mxDeclareMathFun{\fDist}{\mathrm{Dist}}
\mxDeclareMathFun{\fPowerSet}{\mathcal{P}}
\mxDeclareMathFun{\fSupp}{\mathrm{Supp}}

\mxDeclareMathSymbol{\aTime}{t}
\mxDeclareMathSymbol{\arrivalTime}{t}[{^\downarrow}\mskip-1mu]
\mxDeclareMathSymbol{\occurrenceTime}{t}[{^\uparrow}\mskip-1mu]

\mxDeclareMathFun{\fActiveRuns}{\mathfrak{R}_{\mathbf{A}}}
\mxDeclareMathFun{\fActivePrefixRuns}{\mathfrak{R}^{\mathrm{pref}}_{\mathbf{A}}}

\mxDeclareMathSymbol{\anIndex}{i}
\mxDeclareMathSymbol{\anOtherIndex}{j}

\mxDeclareMathSymbol{\aHypothesis}{\mathcal{H}}

\mxDeclareMathSymbol{\historyBound}{B}

\mxDeclareMathFun{\fMinLatency}{\mathbf{l}_\mathrm{min}}
\mxDeclareMathFun{\fMaxLatency}{\mathbf{l}_\mathrm{max}}

\mxDeclareMathSymbol{\anAoA}{\theta}
\mxDeclareMathFun{\fRemapping}{r}

\newcommand{\failStuckAct}[1]{\texttt{fail\_stuck}}

\mxDeclareMathFun{\fArriveAt}{\Omega}
\mxDeclareMathFun{\fEventsArriveAt}{\mathbf{E}}
\mxDeclareMathSymbol{\arrivalTimes}{T_a}
\mxDeclareMathFun{\fIntSysObs}{I}
\mxDeclareMathFun{\fRemappingInterval}{\bar{r}}

\mxDeclareMathSymbol{\aFault}{\mathbf{f}\mskip-1mu}

\mxDeclareMathFun{\fDiag}{\mathcal{D}}

\mxDeclareMathFun{\fTProjections}{\mathbf{P}}
\mxDeclareMathFun{\fTFuture}{\mathbf{C}}

\mxDeclareMathSymbol{\obsHistory}{H}

\mxDeclareMathFun{\fActiveSubsets}{\mathbf{A}}

\mxDeclareMathSymbol{\diagDelay}{\Delta}

\mxDeclareMathSymbol{\aDriftParam}{\epsilon}

\mxDeclareMathFun{\fRemappedObsTime}{I}[][][][o]

\mxDeclareMathSymbol{\anObservation}{\omega}
\mxDeclareMathSymbol{\anObservationSet}{\Omega}

\mxDeclareMathBin{\until}{\mathbf{U}}

\mxDeclareMathBin{\parallel}{\mid\mid}

\mxDeclareMathBin{\ltlSatisfies}{\underset{\textsc{ltl}}{\vDash}}
\mxDeclareMathBin{\ctlSatisfies}{\underset{\textsc{ctl}}{\vDash}}

\mxDeclareMathBin{\happensBefore}{\prec}
\mxDeclareMathBin{\notHappensBefore}{\not\prec}

\mxDeclareMathSymbol{\sDiagnosisStates}{\mathbf{D}}
\mxDeclareMathSymbol{\fDiagnosisState}{\mathbf{d}}

\mxDeclareMathSymbol{\aTransition}{\mathcal{T}}

\mxDeclareMathSymbol{\aTPS}{\mathfrak{T}}

\mxDeclareMathSymbol{\aPTADist}{\mu}
\mxDeclareMathSymbol{\aTSDist}{\lambda}

\mxDeclareMathBin{\distCompose}{\otimes}

\mxDeclareMathSymbol{\anAbstractValuation}{\eta^{\mbox{\smaller[3]\#}}\mskip-2mu}
\mxDeclareMathSymbol{\initialZone}{\bar{\eta}^{\mbox{\smaller[3]\#}}\mskip-2mu}
\mxDeclareMathSymbol{\explainedObservations}{\anObservationSet}
\mxDeclareMathSymbol{\occurredFaults}{\mathbf{F}}

\newcommand{\hhDeclareMathSymbol}[2]{\mxDeclareMathSymbol{#1}{#2}}
\hhDeclareMathSymbol{\fault}{\aFault}

\mxDeclareMathFun{\fExt}{}|\llbracket||\rrbracket|

\mxDeclareMathFun{\fConsistentFragments}{\hat{\rho}}

\mxDeclareMathFun{\fFailures}{\mathit{\mathsf{failures}}}

\mxDeclareMathFun{\zoneReset}{\textsc{reset}}
\mxDeclareMathFun{\zoneFuture}{\textsc{future}}

\mxDeclareMathFun{\fConcretize}{\gamma}

\mxDeclareMathSymbol{\allObservations}{\mathfrak{O}}

\mxDeclareMathFun{\diagnosisAction}{\textsc{action}}
\mxDeclareMathFun{\diagnosisExplore}{\textsc{explore}}

\mxDeclareMathSymbol{\anAbstractState}{s^{\mbox{\smaller[3]\#}}\mskip-2mu}

\mxDeclareMathSymbol{\sAbstractStates}{S^{\#}}

\mxDeclareMathSymbol{\sStateLabels}{\mathrm{AP}}

\mxDeclareMathFun{\fBigO}{\mathcal{O}}

\mxDeclareMathFun{\fIh}{O}[][][][]

\mxDeclareMathFun{\fDiffB}{\mathfrak{B}}

\mxDeclareMathSymbol{\anEventSet}{\mathfrak{E}}
\mxDeclareMathSymbol{\fEventToObs}{\mathcal{R}}

\mxDeclareMathSymbol{\timeDiff}{\Delta t}

\mxDeclareMathFun{\fDur}{\mathrm{dur}}

\mxDeclareMathFun{\greenObs}{\mathsf{S}}
\mxDeclareMathFun{\redObs}{\mathsf{R}}
\mxDeclareMathFun{\yellowObs}{\mathsf{Y}}

\mxDeclareMathFun{\existsFinally}{\mathrm{EF}}
\mxDeclareMathFun{\forallGlobally}{\mathrm{AG}}
\mxDeclareMathSymbol{\sysModel}{\mathfrak{S}}
\mxDeclareMathFun{\diagAutomaton}{\mathfrak{D}\mskip-2mu}

\mxDeclareMathSymbol{\aMaybeFragment}{\rho}

\mxDeclareMathSymbol{\daExceeded}{\mathfrak{E}}
\mxDeclareMathSymbol{\locationTrap}{\bot}

\mxDeclareMathSymbol{\labelConsistent}{\textit{\textsf{consistent}}}
\mxDeclareMathSymbol{\labelInconsistent}{\textit{\textsf{inconsistent}}}
\mxDeclareMathSymbol{\labelTrap}{\textit{\textsf{trap}}\mskip2mu}

\mxDeclareMathSymbol{\sFaults}{\mathbf{F}}
\mxDeclareMathFun{\fProject}{\sigma}

\mxDeclareMathFun{\fProject}{\sigma}

\mxDeclareMathSymbol{\aTradObs}{\hat{\sigma}}

\mxDeclareMathSymbol{\remainingObs}{\Omega_R}
\mxDeclareMathSymbol{\seenFaults}{\mathbf{F}}
\mxDeclareMathSymbol{\labelFailure}{\textit{\textsf{fault}}}

\mxDeclareMathFun{\fFrontier}{\mbox{\smaller[1.5]$\mathfrak{F}$}\mskip-2mu}

\mxDeclareMathSymbol{\failureActions}{\Sigma\mskip-2mu_F}
\mxDeclareMathSymbol{\observableActions}{\sActions[O]}

\mxDeclareMathSymbol{\aRun}{\hat{\rho}}

\mxDeclareMathFun{\actAoA}{\textsc{AoA}}

\mxDeclareMathSymbol{\sClocks}{\mathbb{C}}
\mxDeclareMathSymbol{\sValuations}{\mathcal{V}}

\mxDeclareMathSymbol{\aClockSet}{\mathbf{C}}

\mxDeclareMathSymbol{\sRealTime}{\Real[0][+]}

\mxDeclareMathFun{\aValuation}{\eta}

\mxDeclareMathSymbol{\aTA}{\mathsc{ta}}

\mxDeclareMathSymbol{\aClock}{x}
\mxDeclareMathSymbol{\anOtherClock}{y}

\mxDeclareMathFun{\fLang}{\mathcal{L}}

\mxDeclareMathFun{\aWord}{\hat{\sigma}}

\mxDeclareMathSymbol{\theNullClock}{\mathbf{0}}

\mxDeclareMathBin{\aClockRel}{\sim}

\mxDeclareMathFun{\fClockConstraints}{\mathcal{Z}} 

\mxDeclareMathSymbol{\aGuard}{\mathfrak{a}}
\mxDeclareMathSymbol{\aConstraint}{g}

\mxDeclareMathSymbol{\aLocation}{\ell}

\mxDeclareMathSymbol{\anExecution}{\mathsf{e}}
\mxDeclareMathSymbol{\anExecutionFragment}{\hat{\rho}}

\mxDeclareMathSymbol{\sActions}{\mc{A}}
\mxDeclareMathSymbol{\anAction}{\alpha}
\mxDeclareMathSymbol{\sLocations}{L}

\mxDeclareMathSymbol{\theInitialLocation}{\bar{\ell}}
\mxDeclareMathFun{\theInvFun}{I}
\mxDeclareMathSymbol{\theEdgeRel}{\rightsquigarrow}

\mxDeclareMathSymbol{\anEvent}{e}

\mxDeclareMathSymbol{\aWordFragment}{\hat{\pi}}
\mxDeclareMathFun{\fEvents}{\mathrm{events}}
\mxDeclareMathFun{\fPref}{\mathrm{pref}}
\mxDeclareMathFun{\fLast}{\mathrm{last}}
\mxDeclareMathFun{\fDur}{\mathrm{dur}}

\mxDeclareMathFun{\fObsTime}{I}

\mxDeclareMathSymbol{\aScheduler}{\mathfrak{S}}

\mxDeclareMathFun{\fProb}{P}
\mxDeclareMathBin{\conditionalOn}{\,|\,}

\mxDeclareMathBin{\isConsistentWith}{\,\text{is consistent with}\,}
\mxDeclareMathBin{\isInconsistentWith}{\,\text{is inconsistent with}\,}

\mxDeclareMathFun{\fDirac}{\delta}
\mxDeclareMathFun{\fRuns}{\mathfrak{R}}
\mxDeclareMathFun{\fPlausibleRuns}{\widetilde{\mathfrak{P}}}

\mxDeclareMathFun{\fActive}{\mathrm{active}}

\mxDeclareMathFun{\fTPS}{}|\llbracket||\rrbracket|
\mxDeclareMathFun{\fTS}{}|\llbracket||\rrbracket|

\mxDeclareMathBin{\taReset}{R}

\mxDeclareMathSymbol{\aConfig}{\mathsf{c}}
\mxDeclareMathSymbol{\aConfigGuard}{\mathsf{C}}

\mxDeclareMathSymbol{\validConfigs}{\mathcal{C}}

\mxDeclareMathSymbol{\sGuards}{\mathbf{A}}
\mxDeclareMathFun{\fTrack}{\mathrm{Track}}
\mxDeclareMathFun{\fLabelProj}{\mathrm{Labels}}
\mxDeclareMathFun{\fGuards}{\mathrm{A}}

\mxDeclareMathSymbol{\aVATS}{\mathcal{A}}

\newcommand{\finiteSeqs}[1]{{#1}^{\star}}
\mxDeclareMathSymbol{\seqEmpty}{\epsilon}
\newcommand{\seqConcat}[2]{#1\diamond #2}
\mxDeclareMathSymbol{\aWord}{\sigma}
\mxDeclareMathSymbol{\anAlphabet}{\Sigma}

\mxDeclareMathSymbol{\aTS}{\mathfrak{S}}
\mxDeclareMathSymbol{\tsStates}{\mathcal{S}}
\mxDeclareMathSymbol{\tsActions}{\mathsf{Act}}
\mxDeclareMathSymbol{\tsInitialState}{\bar{s}}
\mxDeclareMathSymbol{\tsInitialStates}{S_I}
\mxDeclareMathSymbol{\tsRel}{T}
\mxDeclareMathSymbol{\tsStateSet}{S}
\mxDeclareMathSymbol{\tsActionSet}{A}
\mxDeclareMathFun{\tsPost}{\Delta}
\mxDeclareMathFun{\tsExec}{S}
\mxDeclareMathFun{\tsExecDet}{s^*\mskip-9mu}
\mxDeclareMathSymbol{\tsState}{s}
\mxDeclareMathSymbol{\tsAction}{\alpha}
\mxDeclareMathSymbol{\tsLabeling}{\mathfrak{L}}
\mxDeclareMathSymbol{\tsLabels}{L}
\mxDeclareMathFun{\tsLabel}{\ell}

\mxDeclareMathFun{\fExecVerdict}{\nu}

\mxDeclareMathSymbol{\tsObsAct}{\smash{\mathsf{OAct}}}
\mxDeclareMathSymbol{\tsUnobsAct}{\smash{\mathsf{UAct}}}
\mxDeclareMathSymbol{\tsFaultAct}{\smash{\mathsf{FAct}}}
\mxDeclareMathSymbol{\fakeVerdict}{\#}

\newcommand{\projection}[2]{\smash{#1 \mid_{#2}}}

\mxDeclareMathFun{\fDiagnosis}{\mathsf{d}}
\mxDeclareMathFun{\aDiagnosis}{\mathsf{D}}
\mxDeclareMathSymbol{\aFaultClassSet}{F}

\mxDeclareMathSymbol{\aLang}{\mathfrak{L}}
\mxDeclareMathSymbol{\aSubAlphabet}{A}

\mxDeclareMathSymbol{\sAtomicPropositions}{\mathsf{AP}}

\mxDeclareMathSymbol{\anAtomicProposition}{p}
\mxDeclareMathSymbol{\anLtlFormula}{\phi}

\mxDeclareMathSymbol{\monTrue}{\texttt{t}}
\mxDeclareMathSymbol{\monFalse}{\texttt{f}}
\mxDeclareMathSymbol{\monUnknown}{\texttt{?}}

\mxDeclareMathFun{\ltlMonSema}{}|\llbracket||\rrbracket|

\newcommand{\emptySet}{\varnothing}

\newcommand{\seqProj}[2]{{#1} |_{#2}}
\newcommand{\proj}[2]{\smash{{#1} |_{#2}}}
\mxDeclareMathFun{\fObsEquiv}{E}
\mxDeclareMathFun{\obsClosure}{X}
\mxDeclareMathSymbol{\initialControlStates}{Q_I}

\mxDeclareMathRel{\sysProd}{\otimes}
\mxDeclareMathRel{\sysProdRel}{\rightsquigarrow}
\mxDeclareMathSymbol{\sysProdSync}{A}

\mxDeclareMathSymbol{\anImpl}{\texttt{I}}
\mxDeclareMathSymbol{\aRefinedVTS}{\mathfrak{R}}

\mxDeclareMathFun{\actP}{\pi_{\tsActions}}
\mxDeclareMathFun{\anP}{\aGuard[][*]\mskip-1mu}

\mxDeclareMathSymbol{\possibilityVerdict}{V}
\mxDeclareMathSymbol{\possibilityProperty}{\psi}

\mxDeclareMathSymbol{\anObsWord}{\rho}

\mxDeclareMathSymbol{\basicEvents}{\mathsf{E}}
\mxDeclareMathSymbol{\boolExpr}{b}

\mxDeclareMathFun{\boolSema}{}|\llbracket||\rrbracket|

\newcommand{\labelSign}{\texttt{sign}}
\newcommand{\labelEncrypt}{\texttt{enc}}
\newcommand{\labelSend}{\texttt{send}}
\newcommand{\featSign}{\mathsf{s}}
\newcommand{\featEncrypt}{\mathsf{e}}
\newcommand{\configEmailBoth}{\aConfig[\featSign{\land}\featEncrypt]}
\newcommand{\configSign}{\aConfig[\featSign]}
\newcommand{\configEncrypt}{\aConfig[\featEncrypt]}

\newcommand{\ie}{i.e.\xspace}
\newcommand{\eg}{e.g.\xspace}
\newcommand{\cf}{cf.\xspace}
\newcommand{\etal}{et al.\xspace}

\newcommand{\faultPump}{\texttt{\textcolor{red!50!black}{pump\_fault}}}
\newcommand{\faultPumpShort}{F_p}
\newcommand{\faultElectric}{\texttt{\textcolor{red!50!black}{short\_circuit}}}
\newcommand{\faultElectricShort}{F_s}
\newcommand{\labelEspresso}{\texttt{request}}

\newcommand{\labelDispense}{\texttt{dispense}}
\newcommand{\labelBurn}{\texttt{burn}}

\newcommand{\Feat}{F}
\newcommand{\ValidFeat}{\mathsf{Conf}}

\newcommand{\anFTS}{\mathcal{F}}

\newcommand{\todoMX}[1]{\textcolor{red}{\textbf{MX:} #1}}
\newcommand{\cdnote}[1]{\textcolor{cyan}{\textbf{CD:} #1}}

\newcommand{\jc}{\texttt{lc}}
\tikzstyle{automaton}=[
  ->, shorten >=1pt, >=stealth',
  node distance=3.5cm,
  every state/.style={minimum size=5mm,fill=white,inner sep=1pt,rounded rectangle},
  ok/.style={fill=white!80!ForestGreen},
  faulty/.style={fill=white!80!red},
  initial text=$ $,
]
\tikzstyle{lattice}=[
  <-
]

\newcommand{\configBoth}{\aConfig[\texttt{b}]}

\newcommand{\faultClasses}{\mathcal{F}}
\mxDeclareMathSymbol{\faultClass}{\mathsf{F}}
\mxDeclareMathSymbol{\aDiagnoser}{\mathfrak{D}}

\mxDeclareMathSymbol{\ltlThreeDomain}{\mathbb{B}_3}
\mxDeclareMathSymbol{\ltlFourDomain}{\mathbb{B}_4}
\mxDeclareMathSymbol{\ltlFiveDomain}{\mathbb{B}_5}

\newenvironment{eqinline}[0]{\vspace{-2pt}\begin{center}$}{$\vspace{-2pt}\end{center}}

\renewcommand{\cdnote}[1]{}

\renewcommand{\todoMX}[1]{}

\begin{document}

\title{
  Configuration Monitor Synthesis\thanks{
    This work was partially supported by the DFG under the projects
    TRR 248 (see {\footnotesize\url{https://perspicuous-computing.science}}, project ID 389792660) and
    EXC 2050/1 (CeTI, project ID 390696704, as part of Germany's Excellence Strategy) and by the 
    NWO through Veni grant VI.Veni.222.431.}
}

\author{
Maximilian A. Köhl%
\orcidlink{0000-0003-2551-2814}%
\inst{1}
\and
Clemens Dubslaff%
\orcidlink{0000-0001-5718-8276}%
\inst{2}
\and
Holger Hermanns
\orcidlink{0000-0002-2766-9615}
\inst{1}
}
\authorrunning{M. A. Köhl, C. Dubslaff, and H. Hermanns}

\institute{%
Saarland University, Saarland Informatics Campus, Saarbrücken, Germany \\
\email{\{koehl,hermanns\}@cs.uni-saarland.de}
\and
Eindhoven University of Technology, Eindhoven, The Netherlands\\
\email{c.dubslaff@tue.nl}
}

\maketitle

\begin{abstract}
The observable behavior of a system usually carries useful information about its internal state, properties, and potential future behaviors.
In this paper, we introduce \emph{configuration monitoring} to determine an unknown configuration 
of a running system based on observations of its behavior.
We develop a \emph{modular} and \emph{generic} pipeline to synthesize automata-theoretic \emph{configuration monitors} from a featured transition system model of the configurable system to be monitored.
The pipeline further allows synthesis under partial observability and network-induced losses as well as predictive configuration monitors taking the potential future behavior of a system into account.
Beyond the novel application of configuration monitoring, we show that our approach also generalizes and unifies existing work on runtime monitoring and fault diagnosis, which aim at detecting the satisfaction or violation of properties and the occurrence of faults, respectively.
We empirically demonstrate the efficacy of our approach with a case study on configuration monitors synthesized from configurable systems community benchmarks.
\end{abstract}

\section{Introduction}
\label{sec:intro}
Almost all modern systems are highly configurable, posing significant challenges for their design, analysis, and maintenance due to the huge amount of possible system configurations~\cite{CzaEis2000,ApeBatKas13a}.
While a system's configuration may be known at time of deployment, this can no longer be assumed at runtime where configurations often are not readily exposed~\cite{WotFriAnd18}.
For instance, configurations of legacy or physical components might be unknown to the running system itself---imagine a factory worker who physically configures a machine prior to a production step.
Configurations might also carry sensitive information and are hence disguised to increase security and privacy~\cite{CasCosMar08}.
Reconfigurations after deployment can be another reason for the configuration being unknown during runtime.
Further, configurations might be not reported due to system faults, \eg, when the system becomes 
unresponsive to queries regarding its configuration.
Since configurations heavily influence a system's properties, knowing them at runtime, however, 
is beneficial for many purposes:
Debugging and other corrective actions might require knowledge about the features a system has 
or how it is configured. 
Tools interacting with the system could adapt their behavior according to the system configuration, \eg, monitoring and diagnostic equipment~\cite{DBLP:conf/isola/DubslaffK22,KimBodBat10}.
Also to detect configuration vulnerabilities~\cite{RamSek02}, to determine information leakages~\cite{PelStrJur18},
or to reason about possible configuration-based attacks that compromise system security, information
about the system configuration is very valuable.

In this paper, we propose a formal and rigorous solution to the problem of determining 
possible configurations of a running system by solely observing its behavior. 
Specifically, we introduce \emph{configuration monitoring},
where system configuration verdicts are drawn from system observations, similar to property verdicts 
that are drawn in runtime monitoring~\cite{DBLP:journals/jlp/LeuckerS09}.

\paragraph{The Challenge: Configuration Monitor Synthesis.}
Configurable systems are typically modeled as \emph{featured transition systems} (FTSs), an extension of transition systems where transitions are guarded by sets of \emph{configurations}~\cite{DBLP:journals/tse/ClassenCSHLR13}.
As an illustrative example, \Cref{fig:email-model} depicts an FTS model of an email system with an encryption ($\featEncrypt$) and signing ($\featSign$) \emph{feature}, of which at least one must be enabled leading to three \emph{valid configurations}~\cite{Kang1990}: $\configEmailBoth$ for both $\featEncrypt$ and $\featSign$, $\configEncrypt$ for just $\featEncrypt$, and $\configSign$ for just $\featSign$.
\begin{wrapfigure}[14]{r}{5.1cm}
  \vspace{-1.6em}
  \begin{tikzpicture}[automaton,font=\footnotesize]
    \node[state,initial left] (initial) {};
    \node[state,below right=2cm of initial] (encrypted) {};
    \node[state,above right=2cm of encrypted] (signed) {};

    \draw (initial) edge[bend left= 20] node[above]{$\{ \configEmailBoth, \configSign \}$ : \labelSign} (signed)
      (signed) edge node[below]{$\{ \configSign \}$ : \labelSend} (initial)
      (signed) edge[bend left=15] node[below,rotate=45]{$\{ \configEmailBoth, \configEncrypt \}$ : \labelEncrypt} (encrypted)
      (encrypted) edge[bend left=50] node[below,rotate=-45]{$\set{\!\configEmailBoth, \configSign, \configEncrypt\!}$ : \labelSend} (initial)
      (initial) edge node[below,rotate=-45]{$\{ \configEncrypt \}$ : \labelEncrypt} (encrypted);
  \end{tikzpicture}
  \caption{Model of an email system with an encryption ($\featEncrypt$) and signing ($\featSign$) feature.}
  \label{fig:email-model}
\end{wrapfigure}
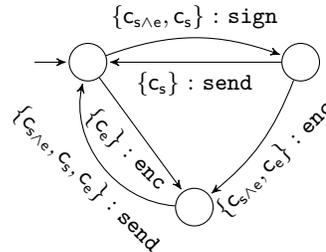
Depending on the configuration, emails are then signed, encrypted, or both before they are sent.
FTSs have been successfully utilized to model and analyze a variety of systems~\cite{Classen2010b,Dev17,DBLP:journals/tse/ClassenCSHLR13}.
While they have been studied in the literature quite extensively, there is no work on effective techniques to determine a given system's configuration by observing its behavior at runtime.
With this paper, we tackle the challenge of synthesizing \emph{configuration monitors} from FTS system models.
When fed with an \emph{observation sequence},
a configuration monitor should output a set of possible configurations.
For instance, a configuration monitor for the email example should output $\set{\configEmailBoth, \configSign}$ after observing $\labelSign$ and $\set{\configEmailBoth}$ after observing $\labelSign$ followed by $\labelEncrypt$.

In practice, determining the possible configurations of a system is aggravated by the fact that systems are usually only \emph{partially observable}, \ie, not all transitions can be observed, and observations may be \emph{lost}, \eg, if they are transmitted over an unreliable network.
The techniques we develope in this paper, allow the effective synthesis of configuration monitors accounting for these \emph{observational imperfections}.
Moreover, they also allow for \emph{predictive configuration monitoring} taking future system behavior into account under the assumption that the system keeps running---a common assumption for reactive systems.

\paragraph{The Approach: Verdict Transition Systems.}
For the effective synthesis of configuration monitors from FTSs,
we develop a \emph{generic} and \emph{modular} synthesis pipeline.
This synthesis pipeline is centered around \emph{verdict transition systems} (VTSs), a novel formalism generalizing \emph{lattice automata}~\cite{DBLP:conf/vmcai/KupfermanL07}.
A VTS represents a system that reads an observation sequence and outputs a \emph{verdict}.
A verdict can be a set of possible systems configurations and VTSs will thus serve as the target representation of configuration monitor synthesis.
The pipeline is generic in the sense that it makes only little assumptions about verdicts, making it useful beyond configuration monitoring.
It is modular in the sense that it consists of multiple building blocks which can be flexibly combined to meet the needs of an application, \eg, in terms of observational imperfections.

Beyond configuration monitoring, we show that VTSs unify and generalize existing work in the spectrum of automata-based runtime monitoring and fault diagnosis (\eg \cite{DBLP:conf/rv/0002LS07,DBLP:conf/fsttcs/BauerLS06,DBLP:journals/tosem/BauerLS11,DBLP:journals/sttt/FalconeFM12,DBLP:journals/tac/SampathSLST95,DBLP:journals/tcst/SampathSLST96}) into a coherent foundation.
While runtime monitoring aims at detecting the satisfaction or violation of properties, diagnosis aims at detecting faults.
Typically, both are integral components of larger systems, serving as mechanisms to pinpoint issues that require corrective action.
This may involve actions like initiating a safe shutdown or alerting a human operator.
As such, they are key elements for ensuring the safe operation of systems, with diverse practical applications (\eg \cite{DBLP:journals/jlp/LeuckerS09,DBLP:conf/rv/KallwiesLSSTW22,DBLP:conf/rv/GorostiagaS18,DBLP:conf/vmcai/BarringerGHS04,DBLP:conf/rv/HuangEZMLSR14}).
This safety critical role also entails the need of dealing with observational imperfections.
We are not aware of any attempts to unify and generalize runtime monitoring and fault diagnosis into a coherent foundation, \eg, such that results and algorithms can be shared and made useable for both.
Through VTSs as a unifying foundation, our pipeline is also useful to synthesize runtime monitors and fault diagnosers under partial observability and tolerant to delays and losses of observations, as well as for producing predictions by taking future system behavior into account.

\paragraph{Contributions and Structure.}
The main contributions of this paper are threefold.
First, we introduce VTSs forming the foundation of our synthesis approach and unifying existing 
work on diagnosis and runtime monitoring (\Cref{sec:foundation}).
Second, we develop modular building blocks for a generic VTS synthesis pipeline and applications to configuration monitoring and beyond (\Cref{sec:synthesis}).
Third, we empirically evaluate our implementation of our framework on configuration monitors synthesized from well-established configurable systems benchmarks (\Cref{sec:evaluation}).

\section{Preliminaries}
\label{sec:preliminaries}
For a finite alphabet $\anAlphabet$ of \emph{symbols}, let $\finiteSeqs{\anAlphabet\mskip-1mu}$ denote the set of finite words over $\anAlphabet$ and
let $\emptySeq \in \finiteSeqs{\anAlphabet\mskip-1mu}$ denote the empty word.
For  $\aWord, \aWord' \in  \finiteSeqs{\anAlphabet\mskip-1mu}$, let $\abs{\aWord}$ denote the length of $\aWord$ and $\seqConcat{\aWord}{\aWord'}$ denote the concatenation of $\aWord$ and $\aWord'$.
For $\aSubAlphabet \subseteq \anAlphabet$ and $\aWord \in \finiteSeqs{\anAlphabet\mskip-1mu}$, the \emph{$\aSubAlphabet$-projection} of $\aWord$, denoted by $\seqProj{\aWord}{\aSubAlphabet}$, removes all symbols not in $\aSubAlphabet$:
\begin{eqinline}
  \seqProj{\emptySeq}{\aSubAlphabet} := \emptySeq
  \qquad
  \seqProj{\left(\seqConcat{\aWord}{\aSymbol}\right)}{\aSubAlphabet} := \begin{cases}
    \seqConcat{\left(\seqProj{\aWord}{\aSubAlphabet}\right)}{\aSymbol} & \text{ if } \aSymbol \in \aSubAlphabet \\
    \seqProj{\aWord}{\aSubAlphabet} & \text{ otherwise}
  \end{cases}
\end{eqinline}
A \emph{language} $\aLang \subseteq \finiteSeqs{\anAlphabet\mskip-1mu}$ is a set of finite words.
Further, let $\projection{\aLang\!\!}{\aSubAlphabet} := \filter{\seqProj{\aWord}{\aSubAlphabet}}{\aWord \in \aLang}$ denote the \emph{$\aSubAlphabet$-projection} of the language $\aLang$.

A \emph{transition system} (TS) is a tuple $\aTS = \tuple{\tsStates, \tsActions, \tsInitialStates, \tsRel}$ comprising a finite set of \emph{states} $\tsStates$, of \emph{actions} $\tsActions$, and of \emph{initial states} $\tsInitialStates {\subseteq} \tsStates$, and a \emph{transition relation} $\tsRel \subseteq \tsStates {\times} \tsActions {\times} \tsStates$.
For $\tsStateSet {\subseteq} \tsStates$ and $\tsActionSet {\subseteq} \tsActions$, let $\smash{\tsPost(\tsStateSet, \tsActionSet)}$ be the set of $\tsActionSet$-successors of $\tsStateSet$, \ie,
$\smash{\tsPost(\tsStateSet, \tsActionSet) := 
\filter{\tsState' \in \tsStates}{\exists \tsState \in \tsStateSet, \tsAction \in \tsActionSet: \tuple{\tsState, \tsAction, \tsState'} \in \tsRel}}$.
We also write $\tsPost(\tsState, \tsAction)$ for $\tsPost(\set{\!\tsState\!}\!, \set{\!\tsAction\!})$.
We call $\aTS$ \emph{deterministic} iff $\abs{\tsInitialStates} = 1$ and 
 $\forall \tsState {\in} \tsStates, \tsAction {\in} \tsActions: \abs{\tsPost(\tsState, \tsAction)} \leq 1$.
An action $\tsAction {\in} \tsActions$ is \emph{enabled} in a state $\tsState {\in} \tsStates$ iff $\tsPost(\tsState, \tsAction) \neq \emptySet$.
$\aTS$ is \emph{action-enabled} iff every action
is enabled in every state.
For $\aWord {\in} \finiteSeqs{\tsActions}$, let $\tsExec(\aWord)$ denote the set of \emph{$\aWord$-reachable} states:
\begin{eqinline}
  \smash{\tsExec(\seqEmpty) := \tsInitialStates
  \qquad
  \tsExec(\seqConcat{\aWord}{\tsAction}) := \tsPost(\tsExec(\aWord), \set{\!\tsAction\!})}
\end{eqinline}
A word $\aWord {\in} \finiteSeqs{\tsActions}$ is \emph{accepted} by $\aTS$ iff $\tsExec(\aWord) {\neq} \emptySet$.
The \emph{language} $\fLang(\aTS)$ of $\aTS$ is the set of its \emph{traces},
i.e., the set of words accepted by $\aTS$.
An \emph{execution} of $\aTS$ is a sequence $\anExecution = (\tuple{\tsState[i], \tsAction[i], \tsState'[i]})_{i=1}^{n} \in \finiteSeqs{\tsRel}$ of transitions such that $\tsState[i] = \tsState'[i-1]$ for $1 < i \leq n$ and $\tsState[1] \in \tsInitialStates$ if $n > 0$.
The trace of the execution $\anExecution$, denoted by $\fTrace(\anExecution)$, is the sequence $(\tsAction[i])_{i=1}^{n}$ of its actions.
We denote the set of all executions of $\aTS$ by $\fExecs(\aTS)$.

Given a set $\Feat$ of features, a subset $\aConfig\subseteq \Feat$ is called a \emph{configuration}.
Systems can only be configured towards valid configurations $\ValidFeat\subseteq\fPowerSet(\Feat)$~\cite{Kang1990}.
Here, $\fPowerSet$ denotes the power set.
Behaviors of configurable systems are typically modeled as featured transition systems~\cite{DBLP:journals/tse/ClassenCSHLR13}.
Formally, an FTS $\anFTS = \tuple{\tsStates, \tsActions, \tsInitialStates, \tsRel, g}$ is a TS extended with a \emph{guard function} $g\colon \tsRel \to \fPowerSet(\ValidFeat) \setminus \set{\!\emptySet\!}$.
A transition $t \in \tsRel$ in $\anFTS$ can only be taken in systems with a configuration $\aConfig\in g(t)$. Formally, the semantics of an FTS $\anFTS$ for configuration $\aConfig \in \ValidFeat$ is a TS ${\anFTS}|_{\aConfig} := \tuple{\tsStates, \tsActions, \tsInitialStates, {\tsRel}|_{\aConfig}}$ where 
${\tsRel}|_{\aConfig} := \filter{t \in T}{\aConfig \in g(t)}$.

A \emph{join-semilattice} is a partially ordered set $\tuple{\aSet, \sqsubseteq}$ where every two-element subset $\set{\anElement, \anElement'} \subseteq \aSet$ has a least upper bound, denoted by $\anElement \join \anElement'$, and called the \emph{join} of $\anElement$ and $\anElement'$.
Analogously, a \emph{meet-semilattice} is a partially ordered set $\tuple{\aSet, \sqsubseteq}$ where every
subset $\set{\anElement, \anElement'} \subseteq \aSet$ has a greatest lower bound, denoted by $\anElement \meet \anElement'$, and called the \emph{meet} of $\anElement$ and $\anElement'$.
Every finite subset $\aSubset \subseteq \aSet$ of a join- or meet-semilattice $\tuple{\aSet, \sqsubseteq}$ has a join or meet, respectively.
For $\tuple{\aSet, \sqsubseteq}$, we refer to the maximal element as the \emph{top element}, denoted by $\top$, and to the minimal element as the \emph{bottom element}, denoted by $\bot$, if they exist.
A \emph{lattice} is a partially ordered set that is both a join- and a meet-semilattice.
We write $\jc(k)$ for the worst-case time complexity of computing the join/meet of $k$ elements.

Lattice automata, pioneered by Kupferman and Lustig, generalize Boolean acceptance of classical finite automata to the multi-valued setting~\cite{DBLP:conf/vmcai/KupfermanL07}, providing an automata-theoretic foundation for multi-valued reasoning about and verification of systems~\cite{DBLP:journals/fttcs/Kupferman22,10.1007/978-3-540-27836-8_26,10.1007/3-540-48683-6_25,10.1007/978-3-642-54830-7_15}.
For a lattice $\tuple{L, \sqsubseteq}$, a \emph{lattice automaton} (LA) is a tuple $\tuple{L, \Sigma, Q, Q_0, \delta, F}$ where $\Sigma$ is a finite alphabet, $Q$ is a finite set of states, $Q_0\colon Q \to L$, 
$\delta\colon Q \times \Sigma \times Q \to L$, and $F\colon Q \to L$.
A run of a lattice automaton on a word $\smash{w = (a_i)_{i=1}^n \in \finiteSeqs{\Sigma}}$ of length $n$ is a sequence $\smash{r = (q_i)_{i=0}^{n} \in \finiteSeqs{Q}}$ of $n + 1$ states.
Each such pair of a word and a run, induces a value of the lattice $L$:
\begin{eqinline}
  \mathsf{val}(w, r) := Q_0(q_0) \sqcap \big( \bigsqcap_{i=1}^n \delta(q_{i - 1}, a_i, q_i)\big) \sqcap F(q_n)
\end{eqinline}
By joining the values of all runs on a word, a value of $L$ is obtained for each word, \ie, $\mathsf{val}(w) := \sqcup \filter{\mathsf{val}(w, r)}{r \text{ is a run on } w}$.
We may now interpret $\mathsf{val}(w) = \top$ and $\mathsf{val}(w) = \bot$ as clear acceptance and rejection of a word $w$, respectively.
A lattice automaton is called \emph{simple}, if the image of $Q_0$ and $\delta$ is $\set{\!\top, \bot\!}$.
For simple lattice automata, $Q_0(q) = \top$ marks $q$ as an initial state and $\delta(q, a, q') = \top$ marks the existence of an $a$-labeled transition from $q$ to $q'$.
The theory of lattice automata has been developed in a series of publications~\cite{DBLP:conf/fossacs/HalamishK11,DBLP:conf/atva/HalamishK12,DBLP:conf/birthday/GonenK15}.

\section{Theoretical Foundation}
\label{sec:foundation}
As the underlying formal foundation of our synthesis approach, we introduce 
\emph{verdict transition systems} (VTSs), a generalization of lattice automata with a focus on system behaviors rather than language acceptance.
VTSs capture how \emph{verdicts} are obtained and evolve over time as new observations are made.

\begin{definition}
  \label{def:vts}
  Let $\tuple{\sVerdicts\!, \isMoreSpecific}\mskip-2mu$ be a join-semilattice, called \emph{verdict domain}.
  A VTS~$\,\aVTS$ over $\mskip2mu\sVerdicts$ is a tuple $\tuple{\sControlStates, \tsActions, \initialControlStates\!, \tsRel\!, \sVerdicts\!, \fVerdict}$ where $\tuple{\sControlStates, \tsActions, \initialControlStates\!, \tsRel}$ is a TS and $\fVerdict\colon\sControlStates\rightarrow\sVerdicts$ is a \emph{verdict function} assigning a verdict to each state.
\end{definition}
\begin{lemma}
  \label{remark:generalization}
  VTSs generalize lattice automata.
\end{lemma}

\begin{proof}
  As established by Kupferman and Lustig, every lattice automaton can be \emph{simplified}, \ie, transformed into a simple one~\cite[Theorem 6]{DBLP:conf/vmcai/KupfermanL07}.
  Recall that for simple lattice automata, $Q_0(q) = \top$ marks initial states and $\delta(q, a, q') = \top$ marks the existence of transitions.
  Hence, any simple lattice automaton $\tuple{L, \Sigma, Q, Q_0, \delta, F}$ is trivially transformed to a VTS with $\sControlStates = Q$, $\tsActions = \Sigma$, $\initialControlStates = \filter{q}{Q_0(q) = \top}$, $\tsRel = \filter{t}{\delta(t) = \top}$, $\sVerdicts = L$, and $\fVerdict = F$.
  While the reverse transformation also applies if $\sVerdicts$ is a lattice, VTSs are more general as they only require the verdict domain $\sVerdicts$ to be a join-semilattice. 
  \qed
\end{proof}

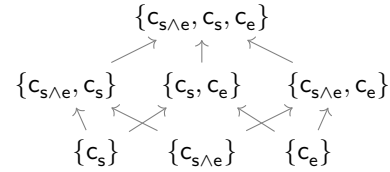
\begin{wrapfigure}[9]{r}{5.4cm}
  \vspace{-2em}
\begin{tikzpicture}[lattice,node distance=.5cm]
      \node (valid) {$\set{\!\configEmailBoth, \configSign, \configEncrypt\!}$};
      \node (v1)[below=.3cm of valid,xshift=1cm,anchor=north west] {$\set{\!\configEmailBoth, \configEncrypt\!}$};
      \node (v2)[below=.3cm of valid,xshift=-1cm,anchor=north east] {$\set{\!\configEmailBoth, \configSign\!}$};
      \node (v3)[below=.3cm of valid] {$\set{\!\configSign, \configEncrypt\!}$};

      \node (x1)[below=.3cm of v1,xshift=-.4cm] {$\set{\!\configEncrypt\!}$};
      \node (x2)[below=.3cm of v2,xshift=.4cm] {$\set{\!\configSign\!}$};
      \node (x3)[below=.3cm of v3] {$\set{\!\configEmailBoth\!}$};


      \path (valid) edge[gray] (v1)
            (valid) edge[gray] (v2)
            (valid) edge[gray] (v3)
            (v1) edge[gray] (x1)
            (v1) edge[gray] (x3)
            (v3) edge[gray] (x1)
            (v2) edge[gray] (x3)
            (v3) edge[gray] (x2)
            (v2) edge[gray] (x2);

\end{tikzpicture}  
  \vspace{-6pt}
  \caption{Configuration verdict domain for the email system (\Cref{fig:email-model}).}
  \label{fig:email-verdict-domain}
\end{wrapfigure}
The partial order $\isMoreSpecific$ of a verdict domain is assumed to order \emph{verdicts} $\aVerdict \in \sVerdicts$ based on their \emph{specificity}.
A verdict $\aVerdict[1]$ is said to be \emph{more specific than} a verdict $\aVerdict[2]$ iff $\aVerdict[1] \isMoreSpecific \aVerdict[2]$ and $\aVerdict[1] \neq \aVerdict[2]$.
Given two verdicts $\aVerdict[1], \aVerdict[2] \in \sVerdicts$, their join $\aVerdict[1] \sqcup \aVerdict[2]$ is the \emph{most-specific} verdict subsuming $\aVerdict[1]$ and $\aVerdict[2]$.
A \emph{configuration verdict} $\aConfigGuard \in \fPowerSet(\ValidFeat) \setminus \set{\!\emptySet\!}$ is a non-empty set of possible configurations
in which a system can exhibit the observed behaviors.
Naturally, a configuration verdict $\aConfigGuard[1]$ is more specific than another $\aConfigGuard[2]$ iff it considers less configurations possible, \ie, iff $\aConfigGuard[1] \subsetneq \aConfigGuard[2]$. 
\Cref{fig:email-verdict-domain} depicts the verdict domain for the email example (\Cref{fig:email-model}).
Thus, configuration monitors can be represented by VTSs over the \emph{configuration verdict domain} $\tuple{\fPowerSet(\ValidFeat) \setminus \set{\!\emptySet\!}, \subseteq}$:

\begin{definition}
  Given a set $\mskip2mu\ValidFeat\mskip-2mu$ of valid configurations, a \emph{configuration monitor} is a VTS with verdict domain $\tuple{\fPowerSet(\ValidFeat) \setminus \set{\!\emptySet\!}, \subseteq}$.
\end{definition}

In every state $\aControlState \in \sControlStates$, a VTS yields a verdict $\fVerdict(\aControlState)$.
As a VTS may be non-deterministic, a given trace $\aWord \in \fLang(\aVTS)$ may lead to multiple states with different verdicts.
To account for these possibilities, we leverage the inspiration from lattice automata for the join-semilattice setting.
\begin{definition}
  \label{def:trace-belief-join}
  Let $\aWord \in \fLang(\aVTS)$ be a trace of a VTS $\aVTS = \tuple{\sControlStates, \tsActions, \initialControlStates, \tsRel, \sVerdicts, \fVerdict}$.
  We define the verdict $\fExecVerdict(\aWord)$ \emph{yielded for} $\aWord$ as follows:
  \begin{equation}
    \smash{\fExecVerdict(\aWord) := \textstyle\bigsqcup \filter{\fVerdict(\aControlState)}{\aControlState \in \tsExec(\aWord)}}\vspace{-2pt}
    \label{eq:def:trace-verdict-join}
  \end{equation}
\end{definition}
Recall that $\tsExec(\aWord)$ is the set of $\aWord$-reachable states, which is finite and non-empty for every $\aWord \in \fLang(\aVTS)$.
Thus, the join in \eqref{eq:def:trace-verdict-join} must exist.
The verdict $\fExecVerdict(\aWord)$ can be interpreted as the most-specific verdict subsuming all non-deterministic possibilities.
Take the verdict domain of the email system as an example:
If both $\set{\configSign}$ and $\set{\configBoth}$ are non-deterministically possible verdicts, the VTS yields $\set{\configSign, \configBoth}$ indicating that both configurations are possible.

A configuration monitor $\aVTS$ targeting a configurable system modeled by an FTS $\anFTS$ 
should be both, \emph{sound} and \emph{complete}.
Soundness requires that for all words $\aWord\in\fLang(\aVTS)$ and \emph{verdict configurations} $\aConfig \in \fExecVerdict(\aWord)$ we have that $\aWord$ is a trace in $\anFTS|_{\aConfig}$, i.e., 
a sound configuration monitor never provides a verdict configuration that has no witness in the system model $\anFTS$.
Completeness ensures that for all witnesses there is also a verdict configuration, i.e., a configuration monitor $\aVTS$ is complete iff $\aConfig \in \fExecVerdict(\aWord)$ for every configuration $\aConfig\in\ValidFeat$ and every trace $\aWord \in \fLang(\anFTS|_{\aConfig})$.
Together, soundness and completeness gurarantee that the configuration monitor $\aVTS$ for the FTS $\anFTS$
always provides the most-specific configuration verdict, \ie, the smallest set of configurations that may result in the observed behaviors.

  In case of our email example, after observing $\labelSign$, the verdict provided by a most-specific configuration monitor 
  should not contain the configuration $\configEncrypt$ (soundness) but both,  $\configSign$ and $\configEmailBoth$ (completeness). This verdict indicates that the system has the sign feature $\featSign$ enabled and still leaves the possibility to have the encryption feature $\featEncrypt$ also enabled.
  If then $\labelSign$ is followed by $\labelSend$, the verdict should be $\set{\!\configSign\!}$, 
  while it should be $\set{\!\configEmailBoth\!}$ if followed by $\labelEncrypt$.

\paragraph{Monotonicity, Refinement, and Equivalence.}
A state $\aControlState \in \sControlStates$ of a VTS is called \emph{monotonic} iff $\fVerdict(\aControlState') \isMoreSpecific \fVerdict(\aControlState)$ for all $\aControlState' \in \tsPost(\set{\!\aControlState\!}\!,\, \tsActions)$.
That is, the verdicts of the state's successors are at least as specific as the verdict of the state itself.
A VTS is called monotonic iff all its states are monotonic.

We further define a refinement and an equivalence relation for VTSs:

\begin{definition}
  \label{def:bts-refinement}
  Let $\aVTS$ and $\aVTS'$ be two VTS over the same verdict domain $\tuple{\sVerdicts, \isMoreSpecific}$.
  We say that $\aVTS$ \emph{refines} $\aVTS'$, denoted by $\aVTS \preceq \aVTS'$,
  iff (i) their language is the same, \ie, $\smash{\fLang(\aVTS) = \fLang(\aVTS')}$, and (ii) $\aVTS$ yields at least as specific verdicts as $\aVTS'$, \ie, $\smash{\fExecVerdict(\aWord) \isMoreSpecific \fExecVerdict'(\aWord)}$ for all $\aWord \in \fLang(\aVTS)$.

  \label{def:vts-equivalence}
  We say that $\aVTS$ and $\aVTS'$ are \emph{verdict-equivalent} iff they refine each other, \ie, iff
  $\,\aVTS \preceq \aVTS'$ and $\aVTS' \preceq \aVTS$.
\end{definition} 
It is easy to see that two verdict-equivalent VTSs  $\aVTS$ and $\aVTS'$ yield exactly the same verdict for each trace, \ie, 
$\fExecVerdict(\aWord) = \fExecVerdict'(\aWord)$ for all $\aWord \in \fLang(\aVTS)$.

\subsection{A Unifying Foundation for Monitoring and Diagnosis?}
\label{sec:unification}
We now demonstrate that VTSs provide a unifying foundation for automata-based monitoring and fault diagnosis.
To this end, we rephrase paradigmatic instances from the literature \cite{DBLP:conf/rv/0002LS07,DBLP:conf/fsttcs/BauerLS06,DBLP:journals/tosem/BauerLS11,DBLP:journals/sttt/FalconeFM12,DBLP:journals/tac/SampathSLST95,DBLP:journals/tcst/SampathSLST96} with VTS concepts.

\begin{figure}[b]
  \centering
  \begin{tikzpicture}[automaton]
    \node[state,initial left] (initial) {$\mathsf{i}$};
    \node[state,right of=initial] (dispensing) {$\mathsf{d}$};
    \node[state,above right=4cm of dispensing, yshift=-3cm] (broken) {$\mathsf{p}$};
    \node[state,below right=4cm of dispensing, yshift=3cm] (burning) {$\mathsf{s}$};

    \draw (dispensing) edge node[above,rotate=7]{\faultPump} (broken)
          (dispensing) edge node[below,rotate=-7]{\faultElectric} (burning)
          (initial) edge[bend left=8] node[above,align=center]{\labelEspresso } (dispensing)
          (burning) edge[loop right] node[right]{\labelBurn} (burning)
          (broken) edge[loop right] node[right,align=left]{\labelEspresso} (broken)
          (dispensing) edge[bend left=8] node[below]{\labelDispense} (initial);
  \end{tikzpicture}
  \vspace{-4pt}
  \caption{Illustrative transition system model of a coffee machine.}
  \label{fig:coffee-model}
\end{figure}
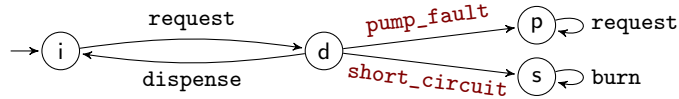

\begin{example}
  As an illustrative example, consider a simple model of a coffee machine (see \Cref{fig:coffee-model}).
  After receiving a coffee request, the machine either dispenses a coffee, or experiences one of two faults:
  Either the pump breaks and it continues accepting requests but ignores them, or there is a short circuit and it stops accepting requests and starts burning.
\end{example}

\paragraph{Diagnosis of Discrete-Event Systems.}
Diagnosis of discrete-event systems is concerned with synthesizing \emph{diagnosers} from system models~\cite{DBLP:journals/tac/SampathSLST95,DBLP:journals/tcst/SampathSLST96}.
To this end, the system is assumed to be modeled as a deterministic TS $\aTS = \tuple{\tsStates, \tsActions, \tsInitialStates, \tsRel}$, whose actions are partitioned into a set of \emph{observable} $\tsObsAct \subseteq \tsActions$ and \emph{unobservable} $\tsUnobsAct \subseteq \tsActions$ actions.
The latter includes a set $\tsFaultAct \subseteq \tsUnobsAct$ of \emph{fault actions}, partitioned into \emph{fault classes} $\smash{\faultClasses = \set{\faultClass[1], \ldots, \faultClass[n]}}$.

Assuming that faults may occur multiple times, a \emph{diagnoser} is a deterministic TS $\aDiagnoser = \tuple{\fPowerSet(\tsStates \times \fPowerSet(\faultClasses))\!, \tsObsAct, \initialControlStates, \tsRel}$.
Each state $\aControlState$
of $\aDiagnoser$ corresponds to a \emph{diagnosis} $\fDiagnosis(\aControlState)  = \filter{\aFaultClassSet}{\tuple{\tsState, \aFaultClassSet} \in \aControlState} \subseteq \fPowerSet(\faultClasses)$.
A diagnosis $\aDiagnosis \subseteq \fPowerSet(\faultClasses)$ is a set of sets of fault classes.
Each $\aFaultClassSet \in \aDiagnosis$ indicates a possibility that faults of the classes $\faultClass[i] \in \aFaultClassSet$ occurred.
Hence, a fault of class $\faultClass[i]$ \emph{certainly occurred} iff $\faultClass[i] \in \aFaultClassSet$ for all $\aFaultClassSet \in \aDiagnosis$ and it \emph{possibly occurred} iff $\faultClass[i] \in \aFaultClassSet$ for some  $\aFaultClassSet \in \aDiagnosis$~\cite[\cf~Def.~6]{DBLP:journals/tac/SampathSLST95}.
These considerations lead to an inherent specificity order: A diagnosis $\aDiagnosis$ is more specific than another diagnosis $\aDiagnosis'$ iff it considers less sets of fault classes possible, \ie, iff $\aDiagnosis \subsetneq \aDiagnosis'$.
As a result, we can cast $\aDiagnoser$ into a VTS $\aVTS[\aDiagnoser]$ over the verdict domain $\tuple{\fPowerSet(\fPowerSet(\faultClasses)), \subseteq}$ with verdict function $\fDiagnosis$ as given above.
Traditionally, a diagnoser $\aDiagnoser$ is constructed such that $\fLang(\aDiagnoser) = \projection{\fLang(\aTS)\!\mskip-2mu}{\tsObsAct}$.
Hence, for each trace $\aWord \in \fLang(\aTS)$ of the diagnosed system,
$\aVTS[\aDiagnoser]$
yields a diagnosis $\fExecVerdict(\seqProj{\aWord}{\tsObsAct})$ as verdict that indicates which faults occurred while taking only observable actions into account.
In general, VTSs obtained from diagnosers may be non-monotonic.
For further details on diagnosis, we refer to Sampath \etal \cite{DBLP:journals/tac/SampathSLST95}.

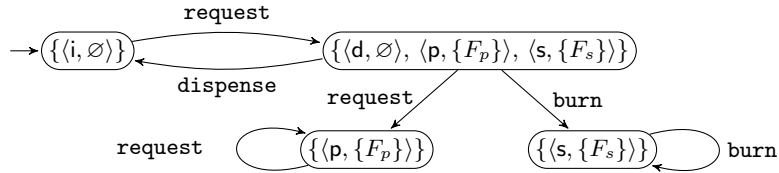
\begin{figure}[t]
  \centering
  \begin{tikzpicture}[automaton]
    \node[state,initial left] (initial) {$\{ \tuple{\mathsf{i}, \emptySet} \}$};
    \node[state,right=2.5cm of initial] (dispense) {$\{\tuple{\mathsf{d}, \emptySet}$, $\tuple{\mathsf{p}, \set{\!\faultPumpShort\!}}$, $\tuple{\mathsf{s}, \set{\!\faultElectricShort\!}}\}$};
    \node[state,below=.8cm of dispense,xshift=1.5cm] (burn) {$\{ \tuple{\mathsf{s}, \set{\!\faultElectricShort\!}} \}$};
    \node[state,below=.8cm of dispense,xshift=-1.5cm] (faulty) {$\{ \tuple{\mathsf{p}, \set{\!\faultPumpShort\!}} \}$};

    \path[->] (initial) edge[bend left=10] node[above,align=center]{$\labelEspresso$} ([yshift=3pt]dispense.west)
        ([yshift=-3pt]dispense.west) edge[bend left=10] node[below]{$\labelDispense$} (initial)
        (dispense) edge
        node[left]{$\labelEspresso$} (faulty)
        (faulty) edge[loop left] node[left,align=right]{$\labelEspresso$\ \ \ } (faulty)
        (dispense) edge node[right]{\ $\labelBurn$} (burn)
        (burn) edge[loop right] node[right]{$\labelBurn$} (burn);
\end{tikzpicture}
\vspace{-6pt}
  \caption{Diagnoser synthesized from the model of the coffee machine (\Cref{fig:coffee-model}).}
  \label{fig:coffee-diagnoser}
\end{figure}

\begin{example}
  \label{ex:coffee-diagnoser}
  \Cref{fig:coffee-diagnoser} depicts the diagnoser constructed from the model of the coffee machine (\Cref{fig:coffee-model}).
  Here, $\labelEspresso, \labelDispense, \labelBurn \in \tsObsAct$ are observable actions and $\faultPump,\faultElectric \in \tsFaultAct$ are fault actions.
  Further, each fault action forms its own fault class $\faultPumpShort$ and $\faultElectricShort$, respectively.
  In states $\set{\tuple{\mathsf{p}, \set{\!\faultPumpShort\!}}}$ and $\set{\tuple{\mathsf{s}, \set{\!\faultElectricShort\!}}}$ of the diagnoser, $\faultPump$ and $\faultElectric$ certainly occurred, respectively, while in state $\set{\tuple{\mathsf{i}, \emptySet}}$ no fault possibly occurred.
\end{example}

\paragraph{LTL Runtime Monitoring.}
\emph{Linear temporal logic} (LTL)~\cite{DBLP:conf/focs/Pnueli77} is used to express properties over infinite words, extending propositional logic with temporal operators.
We use the notation $\ltlEventually \phi$ for \emph{eventually $\anLtlFormula$}, $\ltlAlways \phi$ for \emph{globally $\anLtlFormula$}, and $\ltlNext \anLtlFormula$ for \emph{next $\anLtlFormula$}. 
For an LTL property $\anLtlFormula$, LTL runtime monitoring aims at deciding from a finite prefix of an ongoing run of a system whether $\anLtlFormula$ is satisfied or violated, independent from the future behaviors of a system~\cite{DBLP:conf/fsttcs/BauerLS06,DBLP:journals/tosem/BauerLS11,DBLP:journals/sttt/FalconeFM12,DBLP:conf/rv/0002LS07}.
To adapt the classical LTL semantics to finite prefixes of ongoing runs, different \emph{truth domains} have been proposed \cite{DBLP:journals/sttt/FalconeFM12}.
\LTLThree monitoring uses $\ltlThreeDomain = \set{\monTrue, \monUnknown, \monFalse}$ and, for a prefix $\aWord$ of an ongoing run,
yields $\monTrue$ iff all infinite continuations of $\aWord$ satisfy $\anLtlFormula$, yields $\monFalse$ iff all infinite continuations of $\aWord$ violate $\anLtlFormula$, and yields $\monUnknown$ otherwise \cite[Def.~1]{DBLP:conf/fsttcs/BauerLS06}.
Naturally, $\monTrue$ and $\monFalse$ are more specific than $\monUnknown$ since they represent definite truth values while $\monUnknown$ does not.
Every \LTLThree monitor~\cite[Def.~2]{DBLP:conf/fsttcs/BauerLS06} is a TS which 
can be cast into an action-enabled, deterministic, and monotonic VTS over $\ltlThreeDomain$.

\begin{figure}[t]
  \centering
  \begin{tikzpicture}[automaton]
    \node[state,initial above] (initial) {\texttt{?}};
    \node[state,right=2.3cm of initial] (dispense) {\texttt{?}};
    \node[state,right=2.3cm of dispense] (violation) {\texttt{f}};

    \draw (initial) edge[loop left] node[left,align=right]{$\labelDispense$} (initial)
      (initial) edge[bend left=10] node[above]{$\labelEspresso$} (dispense)
      (dispense) edge[bend left=10] node[below]{$\labelDispense$} (initial)
      (dispense) edge node[above]{$\labelEspresso$} (violation)
      (violation) edge[loop right] node[right,align=left]{$\labelDispense$\\ $\labelEspresso$} (violation);

  \end{tikzpicture}
  \vspace{-4pt}
  \caption{$\text{LTL}_3$ monitor for the LTL property $\anLtlFormula = \ltlAlways\left( \labelEspresso \boolImplies \ltlNext \labelDispense \right)$.}
  \label{fig:rv-ltl-monitor}
\end{figure}

\begin{example}
  \label{ex:rv-ltl-property}
  \Cref{fig:rv-ltl-monitor} depicts an \LTLThree monitor for the property $\anLtlFormula$ that every coffee request is met in the next step, \ie, $\anLtlFormula = \ltlAlways\left( \labelEspresso \boolImplies \ltlNext \labelDispense \right)$.
  Note that we can never be sure that this property is satisfied because it may always be violated by the unknown future.
  Hence, there is no state with verdict $\monTrue$.
\end{example}

\begin{wrapfigure}[9]{r}{2.4cm}
  \vspace{-0.8cm}
  \centering
  \centering
\noindent
\begin{tikzpicture}[lattice,node distance=.35cm]
        \node (unknown) {$?$};
        \node[below left=of unknown,xshift=.1cm] (topc) {$\currentlyTrue$};
        \node[below right=of unknown,xshift=-.1cm] (bottomc) {$\currentlyFalse$};
        \node[below=.3cm of topc] (top) {$\texttt{t}$};
        \node[below=.3cm of bottomc] (bottom) {$\texttt{f}$};

        \path (unknown) edge[gray] 
        (topc)
                (unknown) edge[gray] 
                (bottomc)
                (topc) edge[gray] (top)
                (bottomc) edge[gray] (bottom);
\end{tikzpicture}
\vspace{-.5em}
  \caption{\hfill \LTLFour verdict domain.}
  \label{fig:bool-five-domain}
\end{wrapfigure}
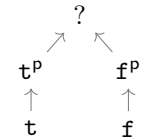
To deal with properties like $\ltlAlways\left( \texttt{request} \boolImplies \ltlEventually \texttt{response} \right)$, \ie, that every request is met eventually, Bauer \etal introduce \LTLFour monitoring using the truth domain $\ltlFourDomain = \set{\monTrue, \currentlyTrue, \currentlyFalse, \monFalse}$, where $\currentlyTrue$ denotes \emph{possibly true} and $\currentlyFalse$ denotes \emph{possibly false} \cite{DBLP:conf/rv/0002LS07,DBLP:journals/sttt/FalconeFM12}.
As such properties can always be satisfied and violated by the unknown future, $\LTLThree$ monitoring would always yield $\monUnknown$.
For further details 
we refer to Falcone \etal \cite{DBLP:journals/sttt/FalconeFM12}.
To technically accommodate \LTLFour monitoring, 
we introduce a fifth verdict $\monUnknown$ to obtain the verdict domain depicted in \Cref{fig:bool-five-domain}.
Without this fifth verdict, $\currentlyTrue$ and $\currentlyFalse$ would not have a least upper bound.
Analogously to \LTLThree, every \LTLFour monitor can be cast into an action-enabled and deterministic VTS.
This VTS may be non-monotonic since  \LTLFour monitors may toggle between $\currentlyTrue$ and $\currentlyFalse$.

\paragraph{A Unifying Foundation.}
We sketched how VTSs can
serve as a unifying foundation for paradigmatic examples of existing work in the spectrum of automata-based monitoring and diagnosis (\eg \cite{DBLP:conf/fdl/Morin-AlloryB06,DBLP:conf/case/AcarS15,DBLP:conf/fmics/ColomboPS08,DBLP:conf/vmcai/BarringerGHS04,DBLP:journals/automatica/CarvalhoMBL13,DBLP:conf/iccad2/MhamdiNDM17,DBLP:journals/jlp/LeuckerS09}).
For approaches not based on automata, like stream-based monitoring (\eg \cite{DBLP:conf/time/DAngeloSSRFSMM05,DBLP:conf/sbmf/ConventHLS0T18,DBLP:conf/cav/BaumeisterFSST20,DBLP:conf/rv/GorostiagaS18}), we leave it to future work to explore how the generic concept of VTSs may transfer.

\section{Generic Building Blocks for VTS Synthesis}\label{sec:synthesis}
The VTS foundations presented in the last section enable us to approach the challenge of configuration monitor synthesis in a generic way:
From a system model annotated with verdicts, synthesize a VTS that yields most-specific verdicts 
when fed with (imperfect) system observations.
For featured transition systems, this would provide sound and complete configuration monitors.
In this section, we establish an automata-based solution providing modular building blocks for a various problem instances.
These blocks enable handling of verdict predictions, realistic observational imperfections, and enable
efficient VTS implementations. Together, they form a generic and modular VTS synthesis pipeline  (\Cref{fig:pipeline})
where each building block maintains most-specificity of generated VTSs.

\begin{figure}[h]
  \centering
\begin{tikzpicture}[font=\footnotesize,node distance=.6cm,shorten >=1pt, >=stealth']
  \node[align=center,draw,minimum width=2.3cm,minimum height=.9cm] (tracking) {Model-Based \\ Construction};
  \node[draw,align=center,minimum height=.9cm] (lookahead) at ($(tracking.east) + (1.6cm,0)$) {Lookahead \\Refinement};
  \node[right=.5cm of lookahead,draw,align=center,minimum height=.9cm] (projection) {Observability \\ Adjustment};
  \node[right=.5cm of projection,align=center,draw,minimum height=.9cm] (finalize) {Finalization};

  \node[below=0pt of tracking,scale=.9] {\Cref{sec:synth-construction}};
  \node[below=0pt of projection,scale=.9] {\Cref{sec:synth-trans}};
  \node[below=0pt of lookahead,scale=.9] {\Cref{sec:lookahead-refinement}};     
  \node[below=0pt of finalize,scale=.9] {\Cref{sec:synth-finalization}};

  \node[draw,circle,inner sep=2pt] (project_q) at ([xshift=-.35cm]lookahead.west)  {};

  \node[left=.3cm of tracking] {$\vphantom{A}\aTS$};
  \node[right=.3cm of finalize] {$\vphantom{A}\anImpl$};
  
  \draw[shorten >=0pt] (tracking) -- (project_q);
  \draw[->] (project_q) -- (lookahead);
  \draw[->] (lookahead) -- (projection);
  \draw[->] (projection) -- (finalize);
  \draw[->] ([xshift=-.35cm]tracking.west) -- (tracking.west);
  \draw[->] (finalize.east) -- ([xshift=.35cm]finalize.east);
  \draw[->] (project_q) -- ([yshift=.65cm]project_q.north) -| (projection.north);
\end{tikzpicture}
\caption{Generic pipeline for model-based VTS synthesis.}
\vspace{-2em}
\label{fig:pipeline}
\end{figure}
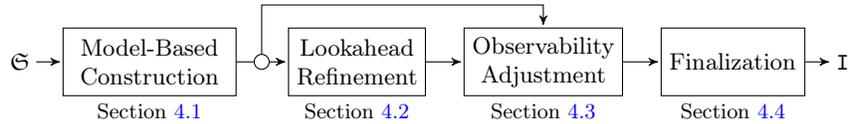

\mxDeclareMathFun{\fStateVerdictAnnotation}{\mathsf{F}}
\mxDeclareMathFun{\fGroundVerdict}{\gamma}
\mxDeclareMathSymbol{\anExecution}{\mathsf{e}}

\subsection{Model-Based VTS Construction}
\label{sec:synth-construction}
Existing constructions for runtime monitors and diagnosers can be cast into VTSs (see \Cref{sec:unification}), however, they are specific to their respective applications and, in the case of monitors, typically do not take a system's model into account.
We here develop a fully generic approach, coined \emph{annotation tracking}, based on \emph{verdict annotations}, which can be used for configuration monitoring as well as runtime monitoring and diagnosis.
Formally, for a TS $\aTS = \tuple{\tsStates, \tsActions, \tsInitialStates\!, \tsRel}$ and a verdict domain $\tuple{\sVerdicts, \sqsubseteq}$, a verdict annotation is a pair of functions $f\colon \tsStates \to \sVerdicts$ and $g\colon \tsRel \to \sVerdicts$ assigning verdicts to states and transitions, respectively.

In the configurable systems domain, FTSs are in fact already TSs whose transitions are annotated with verdicts of the configuration verdict domain $\tuple{\fPowerSet(\ValidFeat){\setminus}\set{\!\emptySet\!}, \subseteq}$ (\cf \Cref{sec:preliminaries}).
For runtime monitoring, a way to generate verdict annotations is to use off-the-shelf model checking~\cite{DBLP:journals/csur/ClarkeW96} to determine whether a state satisfies or violates a given property and then assign a matching verdict from the $\BoolThree$ verdict domain to each state.
For diagnosis, we may assume the transitions with fault actions to be annotated with respective sets of fault classes (we discuss this in more detail later in \Cref{sec:putting-everything-together}).

For each execution $\anExecution \in \fExecs(\aTS)$ of the TS $\aTS$, we define a verdict $\fGroundVerdict(\anExecution)$ by
\bgroup\small
\begin{equation}
  \fGroundVerdict(\emptySeq) := \bigsqcup \filter{f(\tsState)}{\tsState \in \tsInitialStates}
  \quad
  \fGroundVerdict(\left(\tuple{\tsState[i], \tsAction[i], \tsState'[i]}\right)_{i=1}^n) :=
  f(\tsState'[n]) \sqcap \left(\bigsqcap g(\tuple{\tsState[i], \tsAction[i], \tsState'[i]})\right)
  \label{eq:ground-verdict}
\end{equation}
\egroup
if the meet of the respective verdicts exists.
In cases where it does not exist, we ignore the respective execution and leave $\fGroundVerdict$ undefined.
For instance, for FTSs, this means that executions are ignored whose transitions do not share common configurations, as they cannot arise within any valid configuration.

We now aim to construct a VTS $\aVTS$ that yields most-specific verdicts under the idealized assumption that all actions of $\aTS$ can be observed, \ie:
\vspace{-2pt}
\begin{equation}
  \fExecVerdict(\aWord) = \bigsqcup \filter{\fGroundVerdict(\anExecution)}{\anExecution \in \fExecs(\aTS) \text{ with } \aWord = \fTrace(\anExecution) \text{ s.t. } \fGroundVerdict(\anExecution) \text{ is defined}}
  \vspace{-2pt}
  \label{eq:tracking-verdict}
\end{equation}
In words, the VTS $\aVTS$ should produce, for a given trace $\aWord$, the most-specific verdict that subsumes the verdicts $\fGroundVerdict(\anExecution)$ of all executions $\anExecution$ whose trace $\fTrace(\anExecution)$ is $\aWord$ and for which $\fGroundVerdict(\anExecution)$ is defined.
To construct such a VTS, we exploit results from lattice automata~\cite{DBLP:conf/vmcai/KupfermanL07}:
By adding a \emph{sentinel} bottom verdict $\fakeVerdict$ to the verdict domain $\tuple{\sVerdicts, \sqsubseteq}$, we obtain a lattice $L$.
The construction proceeds by first constructing a lattice automaton over $L$, then applying simplification to obtain a simple LA (Theorem~6~\cite{DBLP:conf/vmcai/KupfermanL07}), and finally converting the simple LA to a VTS according to \Cref{remark:generalization} while additionally stripping $\fakeVerdict$-labeled states.
We amalgamate those steps, except the stripping of $\fakeVerdict$-labeled states, into a single definition:

\begin{definition}
  \label{def:annotation-traking-with-sentinel}
  Given the system model $\aTS = \tuple{\tsStates, \tsActions, \tsInitialStates\!, \tsRel}$, the verdict domain $\tuple{\sVerdicts, \sqsubseteq}$, and verdict annotation $f\colon \tsStates \to \sVerdicts$ and $g\colon \tsRel \to \sVerdicts$, we define:
 \begin{eqinline}
    \aVTS[][\fakeVerdict] = \tuple{\tsStates \times \sVerdicts, \tsActions, \tsInitialStates {\times} \set{\!\top\!}, \transitionRel', \sVerdicts\, {\cup} \set{\!\fakeVerdict\!}, \fVerdict}
 \end{eqinline}
 where $\fVerdict(\tuple{\aState, \aVerdict}) = \aVerdict \sqcap f(\aState)$ and $\tuple{\tuple{\aState, \aVerdict}, \tsAction, \tuple{\aState', \aVerdict'}} \in \mathord{\transitionRel'}$ iff there exists a transition $\tuple{\aState, \tsAction, \aState'} \in \transitionRel$ in the system model and $\aVerdict' = \aVerdict \sqcap g(\tuple{\aState, \tsAction, \aState'})$.
\end{definition}

\noindent Note that the meet $\sqcap$ used in the definitions of $\fVerdict$ and $\aVerdict'$ is defined as we extended the verdict domain with $\fakeVerdict$, \ie, it is a lattice.
Stripping $\fakeVerdict$-labeled states from $\aVTS[][\fakeVerdict]$ to obtain the VTS $\aVTS$ over the original verdict domain $\tuple{\sVerdicts, \sqsubseteq}$ is trivial.

\begin{restatable}{theorem}{thmTrackingAccuracy}
  \label{thm:tracking-accuracy}
  For each trace $\aWord \in \fLang(\aVTS)$, we have:
  \vspace{-2pt}
  \begin{equation*}
    \fExecVerdict(\aWord) = \bigsqcup \filter{\fGroundVerdict(\anExecution)}{\anExecution \in \fExecs(\aTS) \text{ with } \aWord = \fTrace(\anExecution) \text{ s.t. } \fGroundVerdict(\anExecution) \text{ is defined}}
    \tag{\ref{eq:tracking-verdict}}
    \vspace{-2pt}
  \end{equation*}
  Furthermore, we have $\fLang(\aVTS) = \filter{\fTrace(\anExecution)}{\anExecution \in \fExecs(\aTS) \text{ s.t. } \fGroundVerdict(\anExecution) \text{ is defined}}$.
\end{restatable}

Recall that FTSs are TSs whose transitions are annotated with verdicts of the domain $\tuple{\fPowerSet(\ValidFeat){\setminus}\set{\!\emptySet\!}, \subseteq}$.
In addition, we assume states to be annotated with $\ValidFeat$, as they are not constrained by configurations.
As per \Cref{thm:tracking-accuracy}, the presented annotation tracking construction produces a sound and complete configuration monitor (\cf \Cref{sec:preliminaries}).

The worst-case time complexity of the construction is $\fBigO(\abs{\sControlStates} \cdot \abs{\sVerdicts} \cdot D \cdot \jc(D))$ where $D$ is the maximal outdegree of $\aTS$.

\paragraph{VTS Specialization.}
Instead of annotating a system model directly, one may also use another action-enabled VTS, \eg, synthesized with third-party techniques, 
and \emph{specialize} it for a given system model.
When combined with the other building blocks, this paves the way for accommodating observational imperfections and enabling predictions based on the model and a VTS obtained with other synthesis techniques (see \Cref{sec:unification}).

\begin{definition}
Let $\aTS = \tuple{\tsStates, \tsActions, \tsInitialStates, \tsRel}$ be a TS and $\aVTS = \tuple{\sControlStates, \sysProdSync, \initialControlStates, \tsRel', \sVerdicts, \fVerdict}$ be an action-enabled VTS with $\sysProdSync\subseteq\tsActions$.
Let $\aVTS[\aTS]\! := \tuple{\tsStates {\times} \sControlStates, \tsActions, \tsInitialStates {\times}\initialControlStates, \sysProdRel, \sVerdicts, \fVerdict'}$ be a VTS with $\fVerdict'(\tuple{\aState, \aControlState}) = \fVerdict(\aControlState)$, and where $\tuple{\tuple{\tsState, \aControlState}, \tsAction, \tuple{\tsState', \aControlState'}} \in \mathbin{\sysProdRel}$ iff (1) $\tsAction \in \sysProdSync$, $\tuple{\tsState, \tsState'} \in \tsRel$, and $\tuple{\aControlState, \aControlState'} \in \tsRel'$, or (2) $\tsAction \not\in \sysProdSync$, $\tuple{\tsState, \tsState'} \in \tsRel$, and $\aControlState = \aControlState'$.
\end{definition}
The VTS $\aVTS[\aTS]$ specializes $\aVTS$ for the system model $\aTS$.
Essentially, $\aVTS[\aTS]$ follows a product construction of $\aVTS$ and $\aTS$, synchronizing over the shared actions $\sysProdSync$.
As $\aVTS$ is action-enabled, the synchronization never blocks and thus, $\fLang(\aVTS[\aTS]) = \fLang(\aTS)$.
Further, $\smash{\fExecVerdict[\aVTS[\aTS]\!](\aWord) = \fExecVerdict[\aVTS](\aWord)}$ for each $\aWord \in \fLang(\aVTS[\aTS])$.
The worst-case time complexity for this specialization construction is $\fBigO(\abs{\sControlStates} \cdot \abs{\tsStates} \cdot \abs{\transitionRel} \cdot \abs{\transitionRel'})$.

\subsection{Most-Specific Predictions}
\label{sec:lookahead-refinement}
Under the additional assumption that the system keeps running, we can further refine the verdicts yielded by a VTS by taking into account possible future behaviors.
In practice, this can be highly valuable for identifying issues earlier, however, it is only justified if the system indeed keeps running.
\emph{Lookahead refinement} refines verdicts of monotonic states $\aControlState \in \sControlStates$ of a VTS 
$\aVTS = \tuple{\sControlStates, \tsActions, \initialControlStates, \tsRel, \sVerdicts, \fVerdict}$
by taking into account future system behaviors starting from $\aControlState$. 
To this end, we define a \emph{lookahead refined-verdict function} $\fVerdict[i]$  for $i \in \Nat$ by:
\begin{eqinline}
  \fVerdict[0](\aControlState) := \fVerdict(\aControlState) \qquad \fVerdict[i + 1](\aControlState) := \begin{cases}
    \bigsqcup_{\aControlState' \in \tsPost(\aControlState)} \fVerdict[i](\aControlState') & \text{if } \forall \aControlState' \in \tsPost(\aControlState): \fVerdict(\aControlState') \isMoreSpecific \fVerdict(\aControlState) \\
    \fVerdict(\aControlState) & \text{otherwise}
  \end{cases}
\end{eqinline}
That is, $\smash{\fVerdict[i + 1]}$ refines the verdict of each monotonic state $\aControlState$ by joining the verdicts of $\aControlState$'s successors from the previous iteration $\smash{\fVerdict[i]}$.
Note that $\fVerdict[i]$ reaches a fixpoint after at most $\smash{\abs{\sControlStates}}$ iterations, ensuring verdicts have propagated from all successors.
Using this fixpoint, we obtain the \emph{lookahead refined} VTS $\aVTS'$:
\begin{eqinline}
  \smash{\aVTS' = \tuple{\sControlStates, \tsActions, \initialControlStates, \tsRel, \sVerdicts, \smash{\fVerdict[\abs{\sControlStates}]}}}
\end{eqinline}
It is easy to see that $\aVTS'$ indeed refines $\aVTS$, \ie, $
\smash{\aVTS' \preceq \aVTS}$ (\cf \Cref{def:vts-equivalence}): The transition relation is not changed and $\smash{\fVerdict[\abs{\sControlStates}]}$ yields more specific verdicts than $\fVerdict$ due to $\fVerdict[i + 1](\aControlState) \isMoreSpecific \fVerdict[i](\aControlState)$.
In particular, we obtain $\fExecVerdict'(\aWord) \isMoreSpecific \fExecVerdict(\aWord)$ for all $\aWord \in \fLang(\aVTS)$.
For the very same reasons, lookahead refinement preserves monotonicity.
Note that $\fExecVerdict'(\aWord)$ is a prediction concerning all possible futures.

Lookahead refinement iterates $\fVerdict[i]$ at most $\abs{\sControlStates}$ times towards a fixpoint,
joining at most $D$ verdicts in time $\jc(D)$ in each iteration, where $D$ is the maximal outdegree of $\aVTS$.
Hence, the worst-case time complexity is $\fBigO(\jc(D) \cdot \smash{\abs{\sControlStates}^2})$.

  \Cref{fig:refine-configs} shows an example of lookahead refinement of configuration verdicts, with an excerpt of the original VTS on the left and its refinement on the right.
  Assume that the system continues running, we know that its configuration must either be $\aConfig[1]$ or $\aConfig[2]$ since $\alpha$ or $\beta$ will inevitably be observed, yielding verdicts $\set{\aConfig[1]}$ or $\set{\aConfig[2]}$, respectively.
  This leads to refining
  $\set{\aConfig[1],\aConfig[2],\aConfig[3]}$ to $\set{\aConfig[1],\aConfig[2]}$.
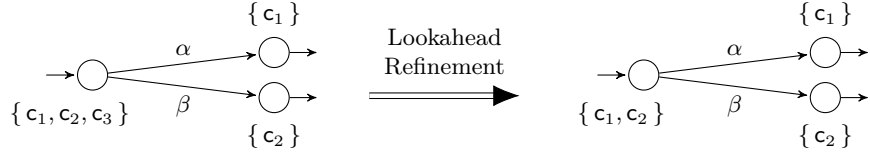
\begin{figure}[h]
  \centering
  \begin{tikzpicture}[automaton,node distance=2cm,font=\footnotesize]
    \node[state,initial left,circle,inner sep=0pt,minimum size=4mm] (initial) {};
    \node[state,right=of initial,inner sep=0pt,yshift=.3cm,minimum size=4mm] (s1) {};
    \node[state,right=of initial,inner sep=0pt,yshift=-.3cm, minimum size=4mm] (s2) {};
    \node[below=1pt of initial,xshift=-.3cm] {$\set{\aConfig[1],\aConfig[2],\aConfig[3]}$};
    \node[above=1pt of s1] {$\set{\aConfig[1]}$};
    \node[below=1pt of s2] {$\set{\aConfig[2]}$};
    \path (initial) edge node[above] {$\alpha$} (s1)
    (initial) edge node[below] {$\beta$} (s2)
    (s1) edge ([xshift=.6cm]s1)
    (s2) edge ([xshift=.6cm]s2);
  \end{tikzpicture}
  \hspace{.4cm}
  \begin{tikzpicture}
    \draw[double,double distance=.6mm,->,arrows = {-Latex[width=8pt, length=10pt]}] (0,0) -> node[above,align=center,yshift=5pt] {Lookahead \\ Refinement} (2,0);
    \draw (0,-.8);
  \end{tikzpicture}
  \hspace{.4cm}
  \begin{tikzpicture}[automaton,node distance=2cm,font=\footnotesize]
    \node[state,initial left,circle,inner sep=0pt,minimum size=4mm] (initial) {};
    \node[state,right=of initial,inner sep=0pt,yshift=.3cm,minimum size=4mm] (s1) {};
    \node[state,right=of initial,inner sep=0pt,yshift=-.3cm, minimum size=4mm] (s2) {};
    \node[below=1pt of initial,xshift=-.3cm] {$\set{\aConfig[1],\aConfig[2]}$};
    \node[above=1pt of s1] {$\set{\aConfig[1]}$};
    \node[below=1pt of s2] {$\set{\aConfig[2]}$};
    \path (initial) edge node[above] {$\alpha$} (s1)
    (initial) edge node[below] {$\beta$} (s2)
    (s1) edge ([xshift=.6cm]s1)
    (s2) edge ([xshift=.6cm]s2);
  \end{tikzpicture}
  \caption{An example of lookahead refinement of configuration verdicts.}
  \vspace{-2em}
  \label{fig:refine-configs}
\end{figure}
\subsection{Observational Imperfections}\label{sec:synth-trans}
In real-world scenarios, observations are rarely perfect.
Accommodating common and often unavoidable observational imperfections, we develop generic VTS transformations for dealing with partial observability, delays, and losses, all in a provably most-specific manner.
In the following, we define such transformations for given a VTS $\aVTS = \tuple{\sControlStates, \tsActions, \initialControlStates, \tsRel, \sVerdicts, \fVerdict}$.
For delays and losses, we handle the bounded and unbounded case.
Industrial network stacks often provide guarantees in the form of bounds that can be obtained by analyses (\eg \cite{felser2005real,DBLP:journals/tie/TovarV99,di2012understanding,DBLP:conf/rtss/TindellHW94}).

\paragraph{Observability Projection.}\label{sec:observability-projection}
Under the assumption that only actions $\tsObsAct \subseteq \tsActions$ are observable, \emph{observability} projection transforms a VTS $\aVTS$ into a VTS $\aVTS'$.
Towards the construction of $\aVTS'$, let $\obsClosure[i](\aControlState)$ be the set of states reachable from $\aControlState \in \sControlStates$ by taking at most $i \in \Nat$ unobservable transitions:
\vspace{-2pt}
\begin{equation}
	\label{eq:VTSobsiter}
  \smash{\obsClosure[0](\aControlState) = \set{\aControlState}
  \qquad
  \obsClosure[i + 1](\aControlState) = \obsClosure[i](\aControlState) \cup \tsPost\big(\obsClosure[i](\aControlState)\!,\, \tsActions {\setminus} \tsObsAct\big)}
  \vspace{-2pt}
\end{equation}
$\obsClosure[i]$ reaches a fixpoint after at most $\abs{\sControlStates}$ iterations.
We define
\vspace{-2pt}
\begin{equation}
	\label{eq:VTSobs}
  \smash{\aVTS' := \tuple{\sControlStates, \tsObsAct, \initialControlStates, \tsRel', \sVerdicts, \fVerdict'}
  \quad\text{with}\quad \fVerdict'(\aControlState) := \textstyle \bigsqcup_{\aControlState' \in \obsClosure[\abs{\mskip-2mu\sControlStates\mskip-2mu}](\aControlState)} \fVerdict(\aControlState')}\vspace{-2pt}
\end{equation}
where $\tuple{\aControlState, \anObservable, \aControlState''} \in \mathord{\transitionRel'}$ iff $\tuple{\aControlState', \anObservable, \aControlState''} \in \tsRel$ for some $\aControlState' \in \obsClosure[\abs{\sControlStates}](\aControlState)$.

\begin{restatable}{theorem}{thmObservabilityProjection}
  \label{thm:observability-projection}
  For
  $\tsObsAct \subseteq \tsActions$, we have (i)~for all $\aWord' \in \fLang(\aVTS')$:
  \vspace{-2pt}
  \begin{equation}
    \smash{\textstyle\fExecVerdict'(\aWord') = 
    \bigsqcup \filter{\fExecVerdict(\aWord)}{\text{for } \aWord \in \fLang(\aVTS) \text{ s.t. }\, \proj{\aWord}{\tsObsAct} = \aWord' \vphantom{\bigsqcup}}}
    \vspace{-2pt}
  \end{equation}
  and (ii)~$\fLang(\aVTS') = \proj{\fLang(\aVTS)}{\tsObsAct}$.
\end{restatable} %
\noindent \Cref{thm:observability-projection} states that $\aVTS'$ yields most-specific verdicts subsuming the verdicts generated by traces with the same $\tsObsAct$-projection, \ie, traces that are indistinguishable when only observing actions in $\tsObsAct$.

For constructing $\aVTS'$, the fixpoint computation on $X_i$ (see \eqref{eq:VTSobsiter}) has to be carried out for each of the $\abs{\sControlStates}$ states before at most $\abs{\sControlStates}$ verdicts are joined~\eqref{eq:VTSobs}.
Hence, the worst-case time complexity is $\fBigO(\jc(\abs{\sControlStates}) \cdot \abs{\sControlStates} + \smash{\abs{\sControlStates}^2})$.

\mxDeclareMathSymbol{\aLostArrivedWord}{\omega}
\mxDeclareMathSymbol{\sLostArrivedWords}{\Omega}

\paragraph{Delays.}
Another realistic observational imperfection are delays.
Given a bound $B$, observations may be delayed by up to $B$ steps.
A VTS robust to a delay of up to $B$ steps is synthesized following a similar approach to \eqref{eq:VTSobsiter}:
\vspace{-2pt}
\begin{equation}
	\label{eq:bounded-loss-iter}
  \smash{\obsClosure[0](\aControlState) = \set{\aControlState}
  \qquad
  \obsClosure[i + 1](\aControlState) = \obsClosure[i](\aControlState) \cup \tsPost\big(\obsClosure[i](\aControlState)\!,\, \tsActions\big)}
  \vspace{-2pt}
\end{equation}
If observations may be delayed by up to $B$ steps, we must look up to $B$ observations ahead.
This is achieved with the following transformation:
\vspace{-2pt}
\begin{equation}
  \label{const:bounded-delays}
  \smash{\aVTS' := \tuple{\sControlStates, \tsActions, \initialControlStates, \tsRel, \sVerdicts, \fVerdict'}
  \quad\text{with}\quad \fVerdict'(\aControlState) := \textstyle \bigsqcup_{\aControlState' \in \obsClosure[B](\aControlState)} \fVerdict(\aControlState')}\vspace{-2pt}
\end{equation}
The transformed verdict function $\fVerdict'$ looks up to $B$ observations ahead.
\begin{theorem}
  \label{th:bounded-delays}
  We have (i) for all $\aWord' \in \fLang(\aVTS')$
  \begin{equation}
    \smash{\textstyle\fExecVerdict'(\aWord') = 
    \bigsqcup \filter{\fExecVerdict(\aWord)}{\text{for } \aWord \in \fLang(\aVTS) \text{ s.t. }\, \exists 0 \leq \Delta \leq B: (\aWord[i])_{i=1}^{\abs{\aWord} - \Delta} = \aWord' \vphantom{\bigsqcup}}}
    \vspace{-2pt}
  \end{equation}
  and (ii)~$\fLang(\aVTS') = \filter{(\aWord[i])_{i=1}^{\abs{\aWord} - \Delta}}{(\aWord[i])_{i=1}^{n} \in \fLang(\aVTS) \text{ and } 0 \leq \Delta \leq B}$.
\end{theorem}
\noindent \Cref{th:bounded-delays} states that the verdict yielded for some trace $\aWord'$ by $\aVTS'$ is the most-specific verdict subsuming the verdicts yielded by $\aVTS$ for those traces of which $\aWord'$ may arise by a delay of up to $B$ steps.
\emph{Unbounded delays} can be handled by using the fixpoint $\obsClosure[\abs{\sControlStates}]$ instead of $\obsClosure[B]$.

\paragraph{Losses.}
It is well-known that in communication networks, packet loss often occurs in bursts.
This insight lead to the established Gilbert-Elliott channel model \cite{6768434,6769369}.
Using this model, one may obtain a bound $B$ on consecutive losses such that the probability for more than $B$ losses is below a certain threshold $p$.
We aim to synthesize a VTS that is robust to up to $B$ consecutive losses.
To formalize such \emph{bounded losses}, we follow a formalization established in the literature on \emph{weakly-hard real-time systems} for consecutive deadline misses~\cite{DBLP:journals/tc/BernatBL01,DBLP:journals/tcad/PazzagliaM22}.

Let $\aLostArrivedWord \in \finiteSeqs{\!\set{\mathsf{L}, \mathsf{A}}\!}$ be a finite sequence over the set $\set{\!\mathsf{L}, \mathsf{A}\!}$ where $\mathsf{L}$ indicates that an observation gets lost and $\mathsf{A}$ indicates that it arrives.
The word $\aLostArrivedWord = (\aLostArrivedWord[i])_{i=1}^{n}$ satisfies the constraint of at most $B \in \Nat$ consecutive losses iff:
\begin{eqinline}
  \forall 1 \leq i \leq j \leq n: \aLostArrivedWord[i] = \aLostArrivedWord[j] = \mathsf{L} \land j - i > B \implies \exists i \leq k \leq j: \aLostArrivedWord[k] = \mathsf{A}
\end{eqinline}
That is, there is at least one arrival between any two losses that are more than $B$ apart.
For $B \in \Nat$, we denote the set of such words by $\sLostArrivedWords[B]$.
Given a trace $\aWord \in \fLang(\aVTS)$ of a VTS and a word $\aLostArrivedWord \in \sLostArrivedWords[B]$ of equal length, \ie, $\abs{\aWord} = \abs{\aLostArrivedWord}$, the \emph{$\aLostArrivedWord$-projection} of $\aWord$, denoted by $\proj{\aWord}{\aLostArrivedWord}$, removes all lost observations:
\begin{eqinline}
  \proj{\emptySeq}{\emptySeq} := \emptySeq
  \qquad
  \proj{\left(\seqConcat{\aWord}{a}\right)}{\left(\seqConcat{\aLostArrivedWord}{x}\right)} := \begin{cases}
    \proj{\aWord}{\aLostArrivedWord} & \text{ iff } x = \mathsf{L} \\
    \seqConcat{\left( \proj{\aWord}{\aLostArrivedWord} \right)}{a} & \text{ otherwise}
  \end{cases}
\end{eqinline}
A VTS robust to at most $B$ consecutive losses is synthesized in a similar way as for bounded delays.
In addition to looking $B$ observations ahead, which potentially have been lost, we also need to adapt the transition relation, as those observations may never arrive.
Thus, using $X$ as defined in \eqref{eq:bounded-loss-iter}, we define
\vspace{-2pt}
\begin{equation}
	\label{eq:bounded-loss-vts}
  \smash{\aVTS' := \tuple{\sControlStates, \tsActions, \initialControlStates, \tsRel', \sVerdicts, \fVerdict'}
  \quad\text{with}\quad \fVerdict'(\aControlState) := \textstyle \bigsqcup_{\aControlState' \in \obsClosure[B](\aControlState)} \fVerdict(\aControlState')}\vspace{-2pt}
\end{equation}
and where $\tuple{\aControlState, \anObservable, \aControlState''} \in \mathord{\transitionRel'}$ iff $\tuple{\aControlState', \anObservable, \aControlState''} \in \tsRel$ for some $\aControlState' \in \obsClosure[B](\aControlState)$.
\begin{restatable}{theorem}{thmBoundedLosses}
  \label{th:bounded-losses}
  We have (i) for all $\aWord' \in \fLang(\aVTS')$
  \vspace{-2pt}
  \begin{equation}
  \label{eq:bounded-loss-spec}
  \smash{\textstyle\fExecVerdict'(\aWord') = 
  \bigsqcup \filter{\fExecVerdict(\aWord)}{\text{for } \aWord \in \fLang(\aVTS) \text{ s.t. }\, \exists \aLostArrivedWord \in \sLostArrivedWords[B], \abs{\aLostArrivedWord} = \abs{\aWord}: \proj{\aWord}{\aLostArrivedWord} = \aWord' \vphantom{\bigsqcup}}}
  \vspace{-2pt}
\end{equation}
  and (ii)~$\fLang(\aVTS') = \filter{\proj{\aWord}{\aLostArrivedWord}}{\aWord \in \fLang(\aVTS) \text{ and } \aLostArrivedWord \in \sLostArrivedWords[B] \text{ s.t. } \abs{\aLostArrivedWord} = \abs{\aWord}}$.
\end{restatable}
\noindent \Cref{th:bounded-losses} states that the verdict yielded for some trace $\aWord'$ by $\aVTS'$ is the most-specific verdict subsuming the verdicts yielded by $\aVTS$ for traces of which $\aWord'$ may arise by up to $B$ consecutive losses, with $B = \abs{\sControlStates}$ handling unbounded losses.

\paragraph{Possiblity Lifting.}
When faced with observational imperfections, certain verdicts become indistinguishable.
We deal with those indistinguishable verdicts by subsuming them into a most-specific verdict (\cf join in \eqref{eq:VTSobs}, \eqref{eq:bounded-loss-vts}, and \eqref{const:bounded-delays}).
It can be beneficial to keep verdicts as individual \emph{possibilities} instead, \eg, when synthesizing diagnosers.
\emph{Possibility lifting} of a VTS $\aVTS$ achieves this by 
replacing its verdict domain $\tuple{\sVerdicts, \sqsubseteq}$ by $\tuple{\fPowerSet(\sVerdicts), \subseteq}$ 
and $\fVerdict$ by $\fVerdict'(\aControlState) := \set{\fVerdict(\aControlState)}$.
After such a lifting, all presented techniques can be applied and retain individual verdicts.

\subsection{Towards Efficient Implementations}
\label{sec:synth-finalization}
To efficiently implement a VTS, whether in software or hardware, it is desirably deterministic and minimal in size.
This is particularly crucial for environments with space limitations, like embedded devices or FPGAs \cite{DBLP:conf/dsd/BodenKBR04,DBLP:conf/fpl/ZhangZLZLB22}.
Determinization and minimization results have been developed for lattice automata~\cite{DBLP:conf/vmcai/KupfermanL07,DBLP:conf/fossacs/HalamishK11}.
For VTSs, they also follow from elementary automata-theoretic results.

\paragraph{Determinization.}
Any VTS can be determinized by adapting the usual power set construction for finite automata \cite{DBLP:journals/ibmrd/RabinS59}.
Let $\aVTS = \tuple{\sControlStates, \tsActions, \initialControlStates, \tsRel, \sVerdicts, \fVerdict}$ be a VTS.
We define a VTS $\fDm(\aVTS) := \tuple{\fPowerSet(\sControlStates), \tsActions, \set{\initialControlStates}, \transitionRel', \sVerdicts, \fVerdict'}$
where\vspace{-2pt}
\begin{equation}
  \smash{\fVerdict'(\aControlStateSet) := \textstyle \bigsqcup_{\aControlState \in \aControlStateSet} \fVerdict(\aControlState)}\vspace{-2pt}
  \label{eq:vts-det-states}
\end{equation}
and $\tuple{\aControlStateSet, \anObservable, \aControlStateSet'} \in \mathord{\transitionRel'}$ iff $\aControlStateSet' = \tsPost(\aControlStateSet, \set{\anObservable})$.
\begin{restatable}{theorem}{thmDeterminization}
  \label{th:determinization}
  For each VTS $\aVTS$, $\aVTS$ and $\fDm(\aVTS)$ are verdict-equivalent.
\end{restatable}
\noindent For each of the $\smash{2^{\abs{\sControlStates}}}$ states of $\smash{\fDm(\aVTS)}$ the join over at most $\abs{\sControlStates}$ verdicts must be computed as per \eqref{eq:vts-det-states}.
Hence, the worst-case time complexity is $\smash{\bigO(\smash{2^{\abs{\sControlStates}}}\cdot \jc(\abs{\sControlStates}))}$.

\paragraph{Minimization.}
For any finite automaton there is a unique minimal deterministic one accepting the same language~\cite{b5eab576-57e1-36aa-bfa3-62b318183436,myhill1957finite}.
This result carries over to VTSs.
\begin{definition}
  A deterministic VTS $\aVTS$ is \emph{minimal} iff all deterministic VTSs that are verdict-equivalent to $\aVTS$ have at least as many states as $\aVTS$.
\end{definition}
\begin{theorem}
  \label{th:vts-minimal}
  For each VTS there is a unique minimal deterministic VTS.
\end{theorem}
In contrast to deterministic finite automata, VTSs may not be action-enabled, \ie, their transition relation may be a partial function.
To see why the results for finite automata carry over to VTSs, assume all missing transitions to end in a non-accepting trap state, while all other states are accepting.
In the minimal VTS, states then correspond to the classes of the coarsest partition where states with distinct verdicts or distinct Myhill-Nerode equivalence classes~\cite{b5eab576-57e1-36aa-bfa3-62b318183436,myhill1957finite} are separated.
These classes guarantee verdict-equivalence as states with different verdicts or a different language belong to different classes.

Typical minimization algorithms for finite automata use partition refinement.
Building upon Hopcroft's earlier work \cite{HOPCROFT1971189}, Valmari and Lehtinen present an algorithm starting with a partition into accepting and non-accepting states \cite{valmari_et_al:LIPIcs:2008:1328}.
Their algorithm is
adapted for VTS minimization by initially partitioning states according to their verdicts.
The time complexity of $\bigO(\abs{\tsRel} \cdot \log \abs{\sControlStates})$ remains unchanged also for VTSs, as
all states may form their own class.

Minimization preserves the language of a VTS $\aVTS$ in accordance with \Cref{def:vts-equivalence}.
If only verdicts matter but not the language, an alternative approach may construct a VTS $\aVTS'$ with $\fExecVerdict'(\aWord) = \fExecVerdict(\aWord)$ for all $\aWord \in \fLang(\aVTS)$ and where $\aVTS'$ is admitted to accept additional words, \ie, $\fLang(\aVTS') \supseteq \fLang(\aVTS)$.
Such a VTS $\aVTS'$ can be even smaller than the minimized $\aVTS$.
Finding the smallest such $\aVTS'$ is a challenge on its own, and beyond the scope of this paper.
As a first step, we simply adapt the minimization algorithm to only split an equivalence class if there are actual transitions hitting inside and outside of a splitter (\cf \cite{HOPCROFT1971189,valmari_et_al:LIPIcs:2008:1328}).
Note that the resulting 
minimization algorithm, referred to as \emph{language-relaxing} in the sequel, is non-deterministic since it depends on the order in which splitters are considered.

\subsection{Putting Everything Together}
\label{sec:putting-everything-together}
We now have all the building blocks of the generic synthesis pipeline for VTS synthesis~(see \Cref{fig:pipeline}).
In the first step, a VTS is constructed based on a system model, either by annotation tracking or by some other means.
Annotation tracking takes a system model annotated with verdicts, \eg, an FTS, and produces a VTS tracking these annotations.
Optionally, lookahead refinement can be applied to enable predictions.
Then, to account for partial observability, delays, or losses, the presented transformations can be applied to obtain a VTS robust to imperfect observations.
Note that these transformations can be cascaded to obtain a VTS that accounts for multiple imperfections in a most-specific manner.
Lastly, the VTS is determinized and minimized to an efficient implementation.

\begin{wrapfigure}[9]{r}{2.2cm}
  \vspace{-2em}
  \centering
  \begin{tikzpicture}[lattice,node distance=.4cm]
      \node (normal) {$\emptySet$};
      \node (pump)[below left=of normal,xshift=.5cm] {$\set{\faultPumpShort}$};
      \node (electric)[below right=of normal,xshift=-.5cm] {$\set{\faultElectricShort}$};
      \node (all)[below=1.2cm of normal] {$\set{\faultPumpShort, \faultElectricShort}$};

    
    
    
      \path (normal) edge[gray] (pump)
            (normal) edge[gray] (electric)
            (pump) edge[gray] (all)
            (electric) edge[gray] (all);


\end{tikzpicture}  
\vspace{-.5em}
  \vspace{-1.2em}
  \caption{\hfill Fault class lattice.}
  \label{fig:fault-lattice}
\end{wrapfigure}
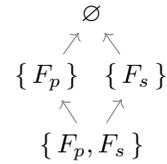
\paragraph{Diagnoser Synthesis.}
We can istantiate the pipeline to synthesize VTSs equivalent to the traditional construction by Sampath \etal \cite{DBLP:journals/tac/SampathSLST95}.
To this end, a TS with fault actions, as described in \Cref{sec:unification}, is annotated with verdicts from $\tuple{\fPowerSet(\faultClasses), \supseteq}$ (\eg, \Cref{fig:fault-lattice}).
Transitions with fault actions $\aFault$ are annotated with singleton sets $\set{\faultClass}$ of the respective fault class $\faultClass \ni \aFault$.
All other transitions (and states) are annotated with
the empty set
$\emptySet$.
Thus, for each execution $\anExecution$, 
the meet in \eqref{eq:ground-verdict} gives us a set of fault classes, corresponding to the faults that occurred on $\anExecution$ since $\sqcap := \cup$.
Using this annotated TS, we obtain a diagnoser by annotation tracking, followed by possibility lifting, observability projection onto $\tsObsAct$, determinization, and then minimization.
Here, possibility lifting changes the verdict domain from $\tuple{\fPowerSet(\faultClasses), \supseteq}$ to the usual diagnosis domain $\tuple{\fPowerSet(\fPowerSet(\faultClasses)), \subseteq}$.
The correctness of the construction follows by combining the theorems of the individual algorithms.
By incorporating lookahead refinement before possibility lifting,
we obtain \emph{predictive diagnosers}.
A predictive diagnoser indicates inevitable faults no later than the traditional techniques.

Moving further beyond the traditional construction, we can also annotate transitions with boolean expressions over some set $\basicEvents$ of \emph{basic fault events} (independent of the actions).
For instance, $e_1 \lor e_2$ or $e_1 \land \lnot e_2$ with $e_1, e_2 \in \basicEvents$.
The usual semantics of boolean expressions induce a lattice $\smash{\tuple{\fPowerSet(\fPowerSet(\basicEvents)), \subseteq}}$ where each expression $\boolExpr$ corresponds to a set $\boolSema(\boolExpr) \subseteq \fPowerSet(\basicEvents)$ of its satisfying assignments.
Naturally, a boolean expression $\boolExpr[1]$ is more specific than another $\boolExpr[2]$ iff $\boolExpr[1]$ implies $\boolExpr[2]$, which is the case iff all assignments satisfying $\boolExpr[1]$ also satisfy $\boolExpr[2]$, \ie, iff $\boolSema(\boolExpr[1]) \subseteq \boolSema(\boolExpr[2])$.
By annotating a TS with boolean expressions indicating whether or not they are enabled in the presence of certain combinations of basic fault events, we obtain an annotated TS over the verdict domain $\smash{\tuple{\fPowerSet(\fPowerSet(\basicEvents)), \subseteq}}$.
Notably, fault trees may serve as a basis for such annotations as they are commonly used to model how top-level faults are caused by lower-level faults, and have a natural interpretation as boolean expressions over a set of basic fault events~\cite{DBLP:journals/csr/RuijtersS15}.
By applying the instance of the VTS synthesis pipeline with possibility lifting as described above to such a TS, we obtain a diagnoser over the verdict domain $\smash{\tuple{\fPowerSet(\fPowerSet(\fPowerSet(\basicEvents))), \subseteq}}$. Here, each verdict can be interpreted as a set of sets of possible worlds in terms of modal logic~\cite{DBLP:books/cu/Chellas80}.
For instance, given a verdict  $\aVerdict \in \fPowerSet(\fPowerSet(\fPowerSet(\basicEvents)))$, it is \emph{necessary} that fault $e_1$ or fault $e_2$ occurred iff $W \subseteq \boolSema(e_1 \lor e_2)$ for all $W \in \aVerdict$, and it is \emph{possible} that fault $e_1$ occurred and fault $e_2$ did not occur iff $W \subseteq \boolSema(e_1 \land \lnot e_2)$ for some $W \in \aVerdict$.
Hence, the thereby constructed diagnoser goes well beyond what is traditionally possible and allows answering powerful modal logic queries.

\paragraph{Related Work.} 
For our theoretical foundation, we generalized lattice automata as introduced by Kupferman and Lustig~\cite{DBLP:conf/vmcai/KupfermanL07}, exploiting their simplification result for annotation tracking.
Zhang, Leucker, and Dong~\cite{DBLP:conf/nfm/ZhangLD12} use finite \emph{predictive words} to obtain possible continuations towards predictive verdicts regarding LTL properties.
Such words can be obtained via static analysis of a monitored program.
The work by Pinisetty \etal~\cite{DBLP:journals/jss/PinisettyJTFMP17} and  Ferrando \etal~\cite{DBLP:journals/fmsd/FerrandoCFLPFM21} incorporates assumptions about a system in terms of properties the system fulfills.
Cimatti, Tian, and Tonetta~\cite{DBLP:conf/rv/CimattiTT19} leverage \emph{fair Kripke structures}~\cite{DBLP:conf/icalp/KestenPR98}, a kind of transition system model, to incorporate assumptions about a system into LTL runtime monitoring.
Their approach can deal with partial observability and produce predictions.
In contrast to the generic techniques we developed, the aforementioned approaches are specifically tailored towards truth verdicts for properties.
Due to the generality of the VTS framework, our approach is also suitable to apply on future automata-based monitoring techniques.

\newcommand{\benchSvm}{\textsc{Svm}\xspace}
\newcommand{\benchMinepump}{\textsc{Minepump}\xspace}
\newcommand{\benchCpterminal}{\textsc{Cpterminal}\xspace}
\newcommand{\benchAerouc}{\textsc{Aerouc5}\xspace}
\newcommand{\benchClaroline}{\textsc{Claroline}\xspace}

\section{Evaluation: Configuration Monitors}
\label{sec:evaluation}
To demonstrate the efficacy of the developed synthesis techniques, we consider configuration monitors synthesized from established FTS benchmarks of the configurable systems community:
\benchSvm and  \benchMinepump~\cite{Classen2010b}, and \benchAerouc, \benchCpterminal, and \benchClaroline~\cite{Dev17}.
\Cref{tb:benchmarks-size} shows an overview of these benchmarks and their important characteristics.
We developed a partially symbolic implementation of the synthesis pipeline instantiated for configuration monitors where we use BDDs~\cite{DBLP:journals/tc/Bryant86,cudd} to succinctly represent and operate on sets of configurations.
There, we first apply annotation tracking to the FTSs, followed by observability projection, determinization, and minimization.
This leads to deterministic, minimal, sound, and complete configuration monitors.
The experiments have been conducted on a 16~core \emph{AMD Ryzen 9 5950X} CPU with 128\,GiB of RAM running Ubuntu~22.04.
An artifact with the implementation and everything necessary for reproducing the experiments is available online \cite{artifact}.

In our evaluation, we aim to answer the following research questions concerning our novel contribution of configuration monitors:
\begin{description}
  \item[RQ1] How do monitor sizes scale with the number of configurations?
  \item[RQ2] What are the potential space savings of minimization?
  \item[RQ3] How does partial observability impact the specificity of verdicts?
\end{description}
\vspace{-.5em}

\begin{table}[t]
  \centering
  \bgroup
  \footnotesize
  \def\arraystretch{1.2}
  \setlength\tabcolsep{3pt}
  \begin{tabular}{r||c|c|c|c|c}
    & \benchSvm &  \benchMinepump & \benchAerouc & \benchCpterminal & \benchClaroline \\ \hline\hline
    $\abs{\ValidFeat}$ & 24 & 32 & 256 & 4\,774 & 820\,193\,280 \\ \hline
    $\abs{\tsActions}$ &  12 & 23 & 11 & 15 & 106 \\ \hline
    FTS & 9/13  & 25/41 & 25/46 & 11/17 & 106/11\,236  \\ \hline\hline
    monitor & 87/120  & 560/992 & 94/178 & 102/161 & 5\,431\,296/575\,717\,376 \\ \hline
    minimized & 87/120  & 496/928 & 56/156 & 93/152 & 65\,536/6\,946\,816\\ 
    relaxed & 17/26  & 53/337 & 4/4& 11/26 & 65\,536/1\,515\,520
  \end{tabular}
  \egroup
  \vspace{6pt}
  \caption{For each model, the rows show (1) the number of valid configurations, (2) the number of actions, (3) the size of the FTS (states/transitions), (4) the size of the monitor constructed from the FTS, and the size of the monitor after (5) language-preserving and (6) language-relaxing minimization.
  }
  \vspace{-2em}
  \label{tb:benchmarks-size}
\end{table}

\paragraph{\textbf{RQ1}: How do monitor sizes scale with the number of configurations?}
Except for \benchClaroline, the size of the FTSs and the number of configurations is comparably small.
Note that a configuration monitor may need to distinguish all possible configurations, potentially leading to an exponential blowup.
\Cref{tb:benchmarks-size} shows the size of the configuration monitors prior to minimization (4) and after minimization (5,6).
We observe across all benchmarks that configuration monitors can be significantly smaller than the number of configurations would suggest.
\benchClaroline shows the greatest divergence with roughly $8 \cdot 10^8$ configurations while the monitor has about $5 \cdot 10^6$ states.
A similar but not as extreme observation can be made about \benchAerouc and \benchCpterminal.
Monitor sizes also influence the construction timings. The \benchClaroline monitor synthesis took around seven minutes, while for all other benchmarks the synthesis (including determinization and minimization) took only few milliseconds.
Reachability analysis on \benchClaroline was already shown to be challenging~\cite{staticanalysis,BeeDamLie21}).
Even for \benchClaroline, our techniques allow for comparably fast and effective configuration monitor synthesis.

\paragraph{\textbf{RQ2}: What are the potential space savings of minimization?}
\Cref{tb:benchmarks-size} shows the size of the monitor after language-preserving (5) and language-relaxing (6) minimization, respectively (see \Cref{sec:synth-finalization}).
For the latter, we also removed self loops, \ie, the monitor stays in its state if a non-enabled action is observed.
In particular, language-relaxing minimization reduces the number of states significantly, leading to very small monitors.
For \benchAerouc, \benchCpterminal, and \benchClaroline, we discover that the number of states is even further reduced.
So, we conclude that minimization can indeed significantly reduce VTS sizes.
Noteworthy, the number of states provides an upper bound on the number of configurations which can be distinguished by observation.
In the extreme case, \benchClaroline, this number is four orders of magnitude lower than the number of configurations.
Thus, most configurations are indistinguishable by an observer, even under full-observability.
So, as a byproduct, our work on configuration monitoring has revealed an explanation for successes reported in family-based analysis~\cite{DBLP:journals/tse/ClassenCSHLR13}.

\newcommand{\dimmed}[1]{\textcolor{gray}{#1}}

\begin{table}[t]
  \centering
  \bgroup
  \small
  \def\arraystretch{1.2}
  \setlength\tabcolsep{4pt}
  \begin{tabular}{r||c|c|c|c||c}
    & $k = 1$ & $k = 2$ & $k = 3$ & $k = 4$ & $k = \abs{\tsActions}$ \\ \hline\hline
    \benchSvm & 26\% \dimmed{(0\%)} & 61\% \dimmed{(23\%)} & 79\% \dimmed{(26\%)} & \textbf{83\%} \dimmed{(33\%)} & 83\% \\ \hline
    \benchMinepump & 26\% \dimmed{(0\%)} & 45\% \dimmed{(0\%)} & 60\% \dimmed{(0\%)} & 71\% \dimmed{(0\%)} & \textbf{79\%} \\ \hline
    \benchAerouc & 25\% \dimmed{(0\%)} & \textbf{44\%} \dimmed{(0\%)} & 44\% \dimmed{(0\%)} & 44\% \dimmed{(0\%)} & 44\% \\ \hline
    \benchCpterminal & 24\% \dimmed{(0\%)} & 37\% \dimmed{(0\%)} & \textbf{40\%} \dimmed{(0\%)} & 40\% \dimmed{(8\%)} & 40\%
  \end{tabular}
  \egroup
  \vspace{6pt}
  \caption{Maximal \dimmed{(minimal)} expected percentage of ruled-out configurations after $1\,000$ steps over all combinations of $k$ observable actions.}
  \vspace{-2em}
  \label{tb:observability}
\end{table}

\paragraph{\textbf{RQ3}: How does partial observability impact the specificity of verdicts?}
To answer this question, we employ the following methodology:
We construct monitors where only a limited number of $k$ actions of $\tsActions$ are considered to be observable. For this, we range over all subsets of $\tsActions$ with $k$ elements, for $1{\leq}k{\leq}4$ and $k{=}\abs{\tsActions}$, and
employ Monte Carlo simulation to compute the expected percentage of ruled-out configurations after $1\,000$ steps.
To this end, $160\cdot10^3$ runs, each of length $1\,000$, are simulated through the system models and observations are fed to the synthesized monitor.
For each of these runs, we uniformly sample a configuration and choose actions uniformly at random.
Note that this gives us a set of expected percentages for each $k$ of which we report the maximum and minimum in \Cref{tb:observability}.
In total, for \Cref{tb:observability}, we synthesized $14\,200$ different monitors and conducted approximately $2\,272\cdot10^6$ simulation runs.
Exploiting parallelization, this took $2.5$ hours on our benchmark machine.
For \benchClaroline this approach is unsuitable as the huge number of valid configurations would require many more runs to obtain statistically significant results.

We see that for \benchAerouc and \benchCpterminal, on average only around 42\% of configurations can be ruled out after $1\,000$ steps (which is sufficient for the monitor to converge on a verdict).
In contrast, for \benchSvm and \benchMinepump, on average 81\% of configurations can be ruled out.
Again, this fits our earlier observation that they have a higher number of states compared to the number of valid configurations than \benchAerouc, \benchCpterminal, and \benchClaroline.

The results also show that the precise set of actions is crucial, as otherwise the specificity of the verdicts may not improve at all.
For instance, for \benchAerouc, with $k = 2$ observable actions, we already can obtain optimal verdicts, while even with $k = 4$, there are sets of observable actions where we cannot discern any configurations at all.
For \benchSvm $k=4$ and for \benchCpterminal $k = 3$ actions can be sufficient.
For \benchMinepump, no $k$-combination for $k \leq 4$ is sufficient.
Note that the number of possible subsets of $\tsActions$ is $2^\abs{\tsActions}$, hence, we did not investigate anything beyond $k = 4$, except $k = \abs{\tsActions}$, which represents full observability.

\section{Conclusion}
\label{sec:conclusion}
With the aim to synthesize configuration monitors, we introduced verdict transition systems (VTSs) as a foundational model which turns out to also cover classical notions of automata-based runtime monitors and diagnosers.

Enabled by this foundation, we developed a modular and generic VTS synthesis pipeline.
We showed that the pipeline can be instantiated for the synthesis of diagnosers, runtime monitors, and configuration monitors, thereby solving the configuration monitor synthesis challenge.
Our techniques go well beyond classical constructions, enabling predictions and robustness to imperfect 
observations for all these applications.
Furthermore, for diagnosis, we showed how to synthesize diagnosers capable of answering powerful 
modal logic fault queries.

We demonstrated the efficacy of our approach for configuration monitoring on multiple well-established benchmarks from the configurable systems community.
Our results show that the approach scales well and can effectively deal with benchmarks which are known to be challenging in the configurable systems literature.
They also provide a new explanation of the success of family-based verification in configurable systems analysis.

For future work, we plan to develop a fully symbolic VTS synthesis framework.
We also plan to extend the observational imperfections, \eg, to bounded reorderings, and study the interplay between different instances of VTSs, \eg, 
in cases where configurations and possible faults interact.
Another interesting direction is to use VTSs to obtain actionable consequences for verdicts 
and as a basis for planning actions towards making verdicts more specific.

\newpage

\bibliographystyle{splncs04}
\bibliography{paper}

\end{document}